\documentclass[11pt]{article}
\usepackage[margin=1in]{geometry}
\usepackage{amsmath,amssymb,amsthm}
\usepackage{graphicx}
\usepackage{booktabs,array,float}
\usepackage{natbib}
\usepackage{enumitem}
\usepackage[hidelinks,hyperfootnotes=false]{hyperref}
\hypersetup{pdftitle={Execution and Evaluation: Task Content, Labor Demand, and the Evaluator Pipeline},pdfauthor={Li Gan}}
\usepackage{setspace}

\newcommand{\eReviewsNineteen}{4,764}
\newcommand{\eRevPerNineteen}{3.0}
\newcommand{\ePapersNineteen}{1,565}
\newcommand{\eWordsNineteen}{404}

\newcommand{\eTotalWordsNineteen}{1.9}

\newcommand{\eConfNineteen}{3.77}

\newcommand{\eRevPerTwentyOne}{3.9}

\newcommand{\eWordsTwentyTwo}{507}

\newcommand{\eConfTwentyThree}{3.63}

\newcommand{\eWordsTwentyFour}{425}

\newcommand{\eReviewsTwentySix}{75,859}

\newcommand{\ePapersTwentySix}{19,474}
\newcommand{\eWordsTwentySix}{439}

\newcommand{\eTotalWordsTwentySix}{33.3}

\newcommand{\eConfTwentySix}{3.62}

\newcommand{\eReviewsRatio}{16}

\newcommand{\hSlopeTriple}{$-$0.2}
\newcommand{\hSlopeTripleSE}{0.6}

\newcommand{\hConfDid}{$-$0.060}
\newcommand{\hConfDidSE}{0.017}

\newcommand{\hConfLastDid}{$-$0.022}
\newcommand{\hConfLastDidSE}{0.023}

\newcommand{\hRhoDid}{0.025}
\newcommand{\hRhoDidSE}{0.025}
\newcommand{\hRhoThPre}{0.55}
\newcommand{\hRhoThPost}{0.38}
\newcommand{\hRhoBePre}{0.56}
\newcommand{\hRhoBePost}{0.37}

\newcommand{\hShareTheoryAll}{12}
\newcommand{\hShareBenchAll}{52}
\newcommand{\hValTheoryLT}{55}
\newcommand{\hValBenchDB}{81}
\newcommand{\hValTheoryDB}{3}
\newcommand{\hValBenchLT}{10}

\newcommand{\iaNatBPre}{-1.35}
\newcommand{\iaNatSPre}{0.36}
\newcommand{\iaNatBPan}{-0.69}
\newcommand{\iaNatSPan}{0.32}
\newcommand{\iaNatBPost}{-1.05}
\newcommand{\iaNatSPost}{0.30}
\newcommand{\iaNatBSix}{-0.88}
\newcommand{\iaNatSSix}{0.24}
\newcommand{\iaFirmBPre}{-0.60}
\newcommand{\iaFirmSPre}{0.27}
\newcommand{\iaFirmPPre}{0.235}
\newcommand{\iaFirmBPan}{0.51}
\newcommand{\iaFirmSPan}{0.34}
\newcommand{\iaFirmPPan}{0.122}
\newcommand{\iaFirmBPost}{-0.50}
\newcommand{\iaFirmSPost}{0.23}
\newcommand{\iaFirmPPost}{0.011}
\newcommand{\iaFirmBSix}{-0.04}
\newcommand{\iaFirmSSix}{0.09}
\newcommand{\iaFirmPSix}{0.679}
\newcommand{\iaEmpBPre}{-0.22}
\newcommand{\iaEmpSPre}{0.31}
\newcommand{\iaEmpPPre}{0.533}
\newcommand{\iaEmpBPan}{0.09}
\newcommand{\iaEmpSPan}{0.33}
\newcommand{\iaEmpPPan}{0.793}
\newcommand{\iaEmpBPost}{-0.24}
\newcommand{\iaEmpSPost}{0.25}
\newcommand{\iaEmpPPost}{0.386}
\newcommand{\iaEmpBSix}{0.00}
\newcommand{\iaEmpSSix}{0.14}
\newcommand{\iaEmpPSix}{0.972}
\newcommand{\iaBtosBPre}{-0.53}
\newcommand{\iaBtosSPre}{0.31}
\newcommand{\iaBtosPPre}{0.181}
\newcommand{\iaBtosBPan}{0.60}
\newcommand{\iaBtosSPan}{0.36}
\newcommand{\iaBtosPPan}{0.091}
\newcommand{\iaBtosBPost}{-0.62}
\newcommand{\iaBtosSPost}{0.24}
\newcommand{\iaBtosPPost}{0.010}
\newcommand{\iaBtosBSix}{0.04}
\newcommand{\iaBtosSSix}{0.14}
\newcommand{\iaBtosPSix}{0.757}
\newcommand{\iaHrAIBPre}{-0.93}
\newcommand{\iaHrAISPre}{0.64}
\newcommand{\iaHrAIPPre}{0.275}
\newcommand{\iaHrAIBPan}{-0.06}
\newcommand{\iaHrAISPan}{0.54}
\newcommand{\iaHrAIPPan}{0.901}
\newcommand{\iaHrAIBPost}{0.66}
\newcommand{\iaHrAISPost}{0.48}
\newcommand{\iaHrAIPPost}{0.295}
\newcommand{\iaHrAIBSix}{0.11}
\newcommand{\iaHrAISSix}{0.34}
\newcommand{\iaHrAIPSix}{0.807}
\newcommand{\iaHrCloudBPre}{0.21}
\newcommand{\iaHrCloudSPre}{0.51}
\newcommand{\iaHrCloudBPan}{0.06}
\newcommand{\iaHrCloudSPan}{0.78}
\newcommand{\iaHrCloudBPost}{-0.38}
\newcommand{\iaHrCloudSPost}{0.47}
\newcommand{\iaHrCloudBSix}{0.38}
\newcommand{\iaHrCloudSSix}{0.63}
\newcommand{\iaHrSoftBPre}{0.17}
\newcommand{\iaHrSoftSPre}{0.14}
\newcommand{\iaHrSoftBPan}{0.63}
\newcommand{\iaHrSoftSPan}{0.73}
\newcommand{\iaHrSoftBPost}{-0.99}
\newcommand{\iaHrSoftSPost}{0.40}
\newcommand{\iaHrSoftBSix}{-0.57}
\newcommand{\iaHrSoftSSix}{0.54}
\newcommand{\iaFirmBDiff}{0.10}
\newcommand{\iaFirmSDiff}{0.39}
\newcommand{\iaFirmPDiff}{0.830}
\newcommand{\iaEmpBDiff}{-0.03}
\newcommand{\iaEmpSDiff}{0.42}
\newcommand{\iaEmpPDiff}{0.953}
\newcommand{\iaHrCloudPPre}{0.743}
\newcommand{\iaHrSoftPPre}{0.686}
\newcommand{\iaHrCloudPPan}{0.938}
\newcommand{\iaHrSoftPPan}{0.366}
\newcommand{\iaHrCloudPPost}{0.364}
\newcommand{\iaHrSoftPPost}{0.048}
\newcommand{\iaHrCloudPSix}{0.644}
\newcommand{\iaHrSoftPSix}{0.531}
\newcommand{\iaCellsPre}{2{,}577}
\newcommand{\iaOccPre}{285}
\newcommand{\iaCellsPan}{2{,}682}
\newcommand{\iaOccPan}{293}
\newcommand{\iaCellsPost}{3{,}227}
\newcommand{\iaOccPost}{341}
\newcommand{\iaCellsSix}{2{,}793}
\newcommand{\iaOccSix}{292}
\newcommand{\iaRelFirm}{0.89}
\newcommand{\iaRelEmp}{0.32}
\newcommand{\iaCorrCloud}{0.92}
\newcommand{\iaCorrSoft}{0.90}
\newcommand{\iaCorrEmpRobot}{0.79}
\newcommand{\iaNSectors}{19}
\newcommand{\iaFirmMin}{1.5}
\newcommand{\iaFirmMax}{7.1}
\newcommand{\iaFirmAbsPost}{0.50}
\newcommand{\iaFirmPostShare}{48}
\newcommand{\iaFineOccBPre}{0.13}
\newcommand{\iaFineOccSPre}{0.15}
\newcommand{\iaFineOccPPre}{0.243}
\newcommand{\iaFineOccBPan}{0.04}
\newcommand{\iaFineOccSPan}{0.14}
\newcommand{\iaFineOccPPan}{0.797}
\newcommand{\iaFineOccBPost}{-0.07}
\newcommand{\iaFineOccSPost}{0.13}
\newcommand{\iaFineOccPPost}{0.597}
\newcommand{\iaFineOccBSix}{0.03}
\newcommand{\iaFineOccSSix}{0.08}
\newcommand{\iaFineOccPSix}{0.754}
\newcommand{\iaFineSecBPre}{0.23}
\newcommand{\iaFineSecSPre}{0.13}
\newcommand{\iaFineSecPPre}{0.203}
\newcommand{\iaFineSecBPan}{0.17}
\newcommand{\iaFineSecSPan}{0.14}
\newcommand{\iaFineSecPPan}{0.016}
\newcommand{\iaFineSecBPost}{-0.14}
\newcommand{\iaFineSecSPost}{0.14}
\newcommand{\iaFineSecPPost}{0.234}
\newcommand{\iaFineSecBSix}{0.07}
\newcommand{\iaFineSecSSix}{0.08}
\newcommand{\iaFineSecPSix}{0.387}
\newcommand{\iaFineCompAIBPre}{0.24}
\newcommand{\iaFineCompAISPre}{0.17}
\newcommand{\iaFineCompAIPPre}{0.323}
\newcommand{\iaFineCompAIBPan}{-0.04}
\newcommand{\iaFineCompAISPan}{0.14}
\newcommand{\iaFineCompAIPPan}{0.771}
\newcommand{\iaFineCompAIBPost}{0.05}
\newcommand{\iaFineCompAISPost}{0.14}
\newcommand{\iaFineCompAIPPost}{0.775}
\newcommand{\iaFineCompAIBSix}{0.09}
\newcommand{\iaFineCompAISSix}{0.09}
\newcommand{\iaFineCompAIPSix}{0.353}
\newcommand{\iaFineCompBPre}{-0.25}
\newcommand{\iaFineCompSPre}{0.22}
\newcommand{\iaFineCompPPre}{0.494}
\newcommand{\iaFineCompBPan}{0.19}
\newcommand{\iaFineCompSPan}{0.21}
\newcommand{\iaFineCompPPan}{0.427}
\newcommand{\iaFineCompBPost}{-0.28}
\newcommand{\iaFineCompSPost}{0.20}
\newcommand{\iaFineCompPPost}{0.076}
\newcommand{\iaFineCompBSix}{-0.14}
\newcommand{\iaFineCompSSix}{0.13}
\newcommand{\iaFineCompPSix}{0.324}
\newcommand{\iaFineBtosOccBPre}{-0.22}
\newcommand{\iaFineBtosOccSPre}{0.20}
\newcommand{\iaFineBtosOccPPre}{0.531}
\newcommand{\iaFineBtosOccBPan}{-0.02}
\newcommand{\iaFineBtosOccSPan}{0.19}
\newcommand{\iaFineBtosOccPPan}{0.909}
\newcommand{\iaFineBtosOccBPost}{-0.16}
\newcommand{\iaFineBtosOccSPost}{0.17}
\newcommand{\iaFineBtosOccPPost}{0.203}
\newcommand{\iaFineBtosOccBSix}{-0.20}
\newcommand{\iaFineBtosOccSSix}{0.09}
\newcommand{\iaFineBtosOccPSix}{0.173}
\newcommand{\iaFineBtosSecBPre}{-0.20}
\newcommand{\iaFineBtosSecSPre}{0.37}
\newcommand{\iaFineBtosSecPPre}{0.307}
\newcommand{\iaFineBtosSecBPan}{-0.46}
\newcommand{\iaFineBtosSecSPan}{0.30}
\newcommand{\iaFineBtosSecPPan}{0.114}
\newcommand{\iaFineBtosSecBPost}{0.03}
\newcommand{\iaFineBtosSecSPost}{0.29}
\newcommand{\iaFineBtosSecPPost}{0.918}
\newcommand{\iaFineBtosSecBSix}{-0.36}
\newcommand{\iaFineBtosSecSSix}{0.18}
\newcommand{\iaFineBtosSecPSix}{0.117}
\newcommand{\iaFineOccBDiff}{-0.20}
\newcommand{\iaFineOccSDiff}{0.18}
\newcommand{\iaFineOccPDiff}{0.420}
\newcommand{\iaFineSecBDiff}{-0.37}
\newcommand{\iaFineSecSDiff}{0.17}
\newcommand{\iaFineSecPDiff}{0.018}
\newcommand{\iaFineCompAIBDiff}{-0.19}
\newcommand{\iaFineCompAISDiff}{0.21}
\newcommand{\iaFineCompAIPDiff}{0.562}
\newcommand{\iaFineCompBDiff}{-0.03}
\newcommand{\iaFineCompSDiff}{0.29}
\newcommand{\iaFineCompPDiff}{0.944}
\newcommand{\iaFineCellsPre}{10{,}690}
\newcommand{\iaFineNoccPre}{315}
\newcommand{\iaFineCellsPan}{11{,}422}
\newcommand{\iaFineNoccPan}{319}
\newcommand{\iaFineCellsPost}{13{,}981}
\newcommand{\iaFineNoccPost}{364}
\newcommand{\iaFineCellsSix}{11{,}782}
\newcommand{\iaFineNoccSix}{319}
\newcommand{\iaFineLowerShare}{31}
\newcommand{\iaFineBtosUnits}{147}
\newcommand{\iaFineUnits}{198}

\newcommand{\iaFineZero}{158}

\newcommand{\iaFineRelABS}{0.77}

\newcommand{\iaFineRelSub}{0.57}
\newcommand{\iaFineRelSubN}{45}
\newcommand{\iaFineRelWithin}{0.21}

\theoremstyle{plain}
\newtheorem{proposition}{Proposition}
\newtheorem{assumption}{Assumption}
\newtheorem{corollary}{Corollary}

\title{Execution and Evaluation: Task Content, Labor Demand,\\ and the Evaluator Pipeline}
\author{Li Gan\thanks{Department of Economics, Texas A\&M University. Email: ganli@tamu.edu; ligan98@gmail.com. I thank colleagues and students for discussions; all errors are mine. Data sources and measure construction are detailed in Section~\ref{sec:data}. Apart from the task scoring described there, AI software (Anthropic's Claude) was used under the author's direction to derive and check the algebra, write the data-processing and estimation code, and edit the draft. All output was checked by the author, who is responsible for the content.}}
\date{October 2026}

\begin{document}
\maketitle

\begin{abstract}
\noindent Measures of occupational AI exposure ask whether machines can perform a task, not which function the human supplies. I distinguish execution, producing a first-pass output, from evaluation, judging whether it is correct, and make the distinction measurable by scoring all $19{,}265$ O*NET task statements under fixed rubrics. The execution score is reproducible across model coders and O*NET vintages. In a $200$-task random audit it correlates $0.81$ with the mean of four blinded human coders. It predicts each held-out coder's peers better than that coder does, though the coders disagree substantially among themselves. It adds only modest power beyond O*NET's evaluation activities. Employment growth is lower in execution-heavy white-collar occupations in every window since 2012; the gradient shows no post-2022 break and lies largely between occupational families. The Eloundou exposure gradient is steeper in 2022--2025 than in the pandemic window, but the design does not establish an AI effect. A near-term seniority test is null. If evaluative judgment is itself produced by execution practice, the supply of evaluators falls only after a cohort lag; that prediction remains untested.
\end{abstract}


\section{Introduction}
\label{sec:intro}

Generative AI has renewed an old question in labor economics: which features of work decide whether a new technology substitutes for workers, complements them, or reorganizes their occupations? Most exposure measures rank tasks by whether a technology can perform them, by their semantic overlap with patents, or by observed use \citep{eloundou2024gpts,webb2020impact,handa2025anthropic}. \citet{massenkoff2026labor} combine capability and use into an ``observed exposure'' index. These measures locate the technological frontier. They do not separately identify the economic function that remains with the human once a machine can produce a first-pass output.

Administrative payroll evidence points to a separate career-stage margin. \citet{brynjolfsson2025canaries} report that workers aged 22--25 in AI-exposed occupations experienced a 16 percent relative employment decline after late 2022. The adjustment came through employment rather than compensation, and it was concentrated in occupations where observed AI use is classified as automating rather than augmenting labor. These facts motivate the paper's seniority tests, but do not by themselves distinguish task substitution from reduced hiring or a depleted evaluator pipeline.

R\'esum\'e data place the same margin inside firms. \citet{hosseini2025seniority} use Revelio Labs records for about 65 million U.S. workers. At firms that posted vacancies for generative-AI integrator roles, junior employment fell by about 9 percent relative to other firms within six quarters of early 2023. Slower hiring, not more separations, drove the decline, and senior employment showed no break. \citet{lambert2026ladder} show that occupational AI exposure is highly correlated with exposure to working from home. When both exposures enter the same hiring regressions, the AI coefficient shrinks sharply and is often indistinguishable from zero, while the remote-work effect remains. The same concern applies to this paper's post-2022 exposure comparisons, which do not control for remote work. These studies classify work by AI exposure or adoption, not by whether it executes or evaluates. They also locate the change in hiring, a margin that the CPS test in Section~\ref{sec:wagemargin} cannot observe.

Recent evidence from China shows why that distinction may matter. Using Zhaopin vacancy descriptions, \citet{zhang2025chinaai} construct an occupation-level LLM exposure index and relate it to declining vacancy shares, greater within-occupation wage dispersion, and higher education and experience requirements. Using matched vacancies, applications, and employer responses from the same platform, \citet{yu2026mismatch} find that after ChatGPT, more-exposed occupations receive more applications but fewer positive responses. Vertical education mismatch falls, and vacancy descriptions specify more and more-complex skills. These studies document immediate labor-demand, screening, and matching responses to one-dimensional AI exposure. They do not distinguish execution from evaluation or observe the cohort-lagged production of evaluative skill, which are the margins at issue here.

I distinguish two functions. \emph{Execution} produces or performs: writing code, drafting a report, calculating an estimate, entering data, or carrying out a procedure. \emph{Evaluation} judges: diagnosing, auditing, reviewing, approving, or deciding whether an output is correct and fit for use. AI may be capable of parts of both, but the functions enter production differently. Execution creates a candidate output. Evaluation supplies verification, certification, responsibility, and judgment. The same occupation can contain both functions, and neither dimension is equivalent to routineness.

The distinction has an intertemporal implication. Evaluation capacity is not necessarily an endowment. Much of it is produced when junior workers perform execution tasks, encounter errors, and learn what a correct output requires. A technology that makes execution cheap can act on verification in opposite directions at different horizons. Candidate volume can rise immediately against a predetermined stock of evaluators. At the same time, automation can remove the practice that produces the next cohort of evaluators. The first effect is contemporaneous verification pressure; the second is delayed pipeline depletion.

The model organizes these channels as a sequence of nested predictions. First, holding AI capability fixed, conditional human labor demand falls more in execution-heavy roles. Second, if entry-level roles contain more execution, their employment response is larger than the response of senior roles. Third, if cheaper execution expands candidate volume, the share that can be checked falls on impact where verification requires scarce human judgment. It does not fall where a mechanical check can absorb the extra volume. Fourth, if execution practice produces evaluative skill, the senior evaluator stock and evaluation quality decline only when the affected entry cohorts mature. The identifying object is not one contemporaneous coefficient. It is the ordered path from execution exposure, to entry adjustment, to immediate verification pressure, and finally to cohort-lagged changes in evaluator quantity and quality.

This paper makes the first link measurable. I score all $19{,}265$ O*NET task statements under a fixed execution rubric and separately score the same statements for AI capability. Constructing both measures on the same task universe permits a direct test of whether execution content matters conditional on what AI can technically do. It also produces a task-level measure that can later be attached to postings, firms, worker histories, and domain-specific corpora.

The evidence comes in five parts.

First, the execution score is reproducible and human-benchmarked (Section~\ref{sec:data}). Rebuilding it on the 2026 O*NET release yields a correlation of $0.996$ with the earlier task universe, and occupation-blind model coders reproduce the task scores at $r=0.924$. In a fixed-seed random audit of $200$ statements, the production score correlates $0.806$ with the mean of four blinded human coders. Holding out each coder and predicting the remaining three, the score outperforms the held-out human in all four cases. I treat the score as human-benchmarked rather than as ground truth. It sits near the center of the human distribution, but that distribution is wide enough that no single correct value exists at the task level.

Second, the execution share is distinct from capability and routineness but overlaps with O*NET's own evaluation activities (Sections~\ref{sec:data} and \ref{sec:onetbenchmark}). Within cognitive occupations, it is only moderately related to AI capability and keeps nearly all of its variation after capability is partialled out. It is positively but imperfectly related to routine-task intensity. O*NET work-activity ratings for evaluating information against standards, monitoring, inspecting, coaching, and consultation form the natural benchmark. The execution share is strongly negatively related to this human-rated composite, as construct validity requires, and the composite predicts 2019--2025 employment growth more strongly. Adding the execution share raises $R^2$ by only $0.006$ in the common cognitive sample and by $0.011$ in the full controlled specification, compared with $0.058$ for the composite. The new measure's value is its common task-level architecture, not a claim that O*NET lacks evaluation information.

Third, employment grew more slowly in execution-heavy cognitive occupations in every window since 2012, with no break after 2022 (Sections~\ref{sec:timing} and \ref{sec:families}). Equality of the execution slopes through 2025 cannot be rejected, and family fixed effects reduce the six-year coefficient from about $-0.70$ to $-0.14$. The gradient is a secular pattern between occupational families, not evidence of an AI-era break. Comparing the same occupation across industries removes the family differences (Section~\ref{sec:adoption}). There, the evidence for an AI-adoption gradient is weak. With the share of firms that used AI in 2017, execution-heavy occupations shrank faster in high-adoption sectors in 2022--2025. The 2016--2019 estimate is similar in size but insignificant, and the 2019--2025 estimate is zero. The employment-weighted adoption measure shows no gradient, sector AI use cannot be separated from software adoption, and no 2022--2025 gradient appears across industries within a sector.

Fourth, the post-2022 evidence dates changes without identifying their cause (Sections~\ref{sec:usage} and \ref{sec:wagemargin}). The Eloundou exposure gradient becomes more negative in 2022--2025 than in the 2019--2022 pandemic window, but a cross-section cannot attribute that steepening. Employment contraction is concentrated in the high-use, high-execution cell, yet the continuous interaction and the family-adjusted super-additivity test are insignificant. A CPS test of entry-age employment is also null, on a margin that omits hiring and occupational entry.

Fifth, one research venue observes both candidate volume and verification capacity (Section~\ref{sec:iclr}). At ICLR, reviews have kept pace with submissions at about four per paper since 2021. Review length fell from its 2022 peak, and reviewer confidence fell. Proof-based and benchmark-based submissions show no differential change in score--citation association, as expected where the venue reads both kinds of paper.

This staged interpretation matters for the paper's relation to existing theory. Prediction-and-judgment models make judgment a distinct costly input \citep{agrawal2018judgment}. Task models allow technologies to substitute for some tasks and complement others \citep{autor2003skill,acemoglu2018race,acemoglu2022tasks}. Recent work makes verification scarcity, apprenticeship, or learning-by-doing central to AI adoption \citep{catalini2026agi,garicano2025training,afrouzi2026automation,dogonowski2026apprenticeship,li2026generation}. A companion application to scientific retractions reports suggestive cross-domain heterogeneity while emphasizing that formal retraction does not identify verification pressure \citep{gan2026verification}. Individual predictions are not unique. The contribution claimed here is narrower: a measurable execution-evaluation axis and a joint empirical sequence that separates the impact response from the cohort-lagged quantity and quality responses. In particular, a decline in evaluation quality per senior worker, conditional on evaluator numbers, requires the practice-to-skill channel rather than cohort accounting alone.

The rest of the paper follows that sequence. Section~\ref{sec:framework} develops execution-share sorting, immediate verification pressure, and the evaluator pipeline, and states which implications current data can examine. Section~\ref{sec:data} constructs and validates the task measures. Section~\ref{sec:results} reports the long-run, within-occupation, post-2022, wage, and seniority facts. Section~\ref{sec:conclusion} separates the evidence in hand from the adoption, hiring, and cohort data needed for causal and delayed-mechanism tests. Appendix~\ref{app:formal} gives the derivations, quantitative extension, and comparison with alternative closures.

\section{Model and Empirical Scope}
\label{sec:framework}

\subsection{Production and execution-share sorting}

A role performs a unit mass of tasks. A fraction $e_j\in[0,1]$ are execution tasks and $1-e_j$ are evaluation tasks. Execution produces a first-pass output. Evaluation determines whether that output is correct, reliable, and fit for use. The second occupation characteristic, $a_j\in[0,1]$, is the share of tasks on which AI is capable. The two dimensions can cross: an AI-capable task may still require a human verdict, while an execution task may be beyond the technology because it requires physical action.

Execution output combines human input $s^E_j$, hired at wage $w$, with AI input $Z_j$, purchased at price $q$, under a constant-elasticity technology,
\begin{equation}
y^E_j=\left[\alpha(s^E_j)^\rho+(1-\alpha)(\theta_j A_t Z_j)^\rho\right]^{1/\rho},
\qquad \sigma=\frac{1}{1-\rho}>1.
\label{eq:Y}
\end{equation}
Frontier capability $A_t$ raises the efficiency of AI input, and role-specific productivity $\theta_j$ is increasing in the occupation's AI-capable share $a_j$. Cost minimization yields conditional human execution demand $s_j^*(A_t)$. Its capability elasticity is
\begin{equation}
\frac{\partial\ln s_j^*}{\partial\ln A_t}=-\sigma(1-\psi^E_j)<0,
\label{eq:psi}
\end{equation}
where $\psi^E_j$ is the human cost share of execution. Appendix~\ref{app:formal} derives the unit cost and equation~\eqref{eq:psi}. The conditional labor quantity falls for any positive substitution elasticity; the maintained case $\sigma>1$ additionally makes the human cost share fall as capability rises.

Evaluation is not assumed to be technologically untouched. The narrower assumption is that some human input remains because certification, liability, and ownership of judgment are not fully delegable.

\begin{assumption}[Oversight residual]
\label{ass:residual}
A fraction $\underline r\in(0,1]$ of evaluation labor remains human at any level of $A_t$.
\end{assumption}

Normalize each execution task to require one unit of execution output and each evaluation task one unit of evaluation labor. Conditional human labor per unit of occupational output is
\begin{equation}
H_j(A_t)=e_j s_j^*(A_t)+(1-e_j)\underline r.
\label{eq:H}
\end{equation}
Equation~\eqref{eq:H} isolates the margin the execution measure is intended to capture. It holds occupational output fixed and treats the oversight residual as common across occupations. Both restrictions matter empirically: cheaper execution can expand output, and licensing or liability can make the retained human share vary systematically across occupations.

\begin{proposition}[Execution-share sorting]
\label{prop:sorting}
Under Assumption~\ref{ass:residual}, both the level decline and the proportional decline of conditional human employment as $A_t$ rises are increasing in $e_j$, holding $a_j$ fixed.
\end{proposition}

\noindent\textit{Proof sketch.} At a fixed AI-capable share, equation~\eqref{eq:H} gives $\partial H_j/\partial A_t=e_j\,\partial s_j^*/\partial A_t$, which becomes more negative as $e_j$ rises. For proportional changes,
\[
\frac{\partial\ln H_j}{\partial A_t}=\underbrace{\frac{e_j s_j^*}{e_j s_j^*+(1-e_j)\underline r}}_{f_j(e_j)}\frac{\partial\ln s_j^*}{\partial A_t}.
\]
The labor share $f_j(e_j)$ is increasing in $e_j$, while the second term is negative by equation~\eqref{eq:psi}. Appendix~\ref{app:formal} gives the full proof.

A numerical illustration shows the size of the wedge. Set $s_j^*=1$ and $\underline r=0.5$. An occupation with $e_j=0.6$ has $f_j=0.6/(0.6+0.2)=0.75$, and one with $e_j=0.4$ has $f_j=0.4/(0.4+0.3)\approx0.57$. A 10 percent fall in human execution demand then cuts total human labor by $7.5$ percent in the first occupation and by $5.7$ percent in the second.

The proposition supplies the model's task-content wedge. A one-dimensional capability measure ranks occupations by whether AI can perform their tasks. It contains no separate execution term once capability is held fixed. Equation~\eqref{eq:H} instead predicts that the same capability improvement has a larger conditional labor effect in a role whose human work is concentrated in execution rather than evaluation. This is a comparison with capability-only exposure measures, not a claim that a general task model cannot incorporate evaluation, regulation, or other task attributes.

\subsection{From the model to the empirical specifications}

Over a window in which frontier capability rises by $\Delta A_t$, the conditional proportional employment change can be written locally as
\begin{equation}
D_j\approx-f_j(e_j)h(a_j)\Delta A_t,
\qquad f_j'>0,\quad h'>0,
\label{eq:Dstar}
\end{equation}
where $h(a_j)$ summarizes the execution-displacement response associated with role-specific AI capability. A first-order empirical approximation is
\begin{equation}
D_j=\beta_0+\beta_e e_j+\beta_a a_j+\beta_{ea}(e_j a_j)+X_j'\gamma+\varepsilon_j.
\label{eq:reg}
\end{equation}
With $D_j$ defined as signed log employment change, the conditional model predicts $\beta_e<0$ and, more sharply, $\beta_{ea}<0$: the execution share should matter more where AI is more capable. The additive coefficient $\beta_a$ is not signed by the reduced form alone. Output expansion, complementarity between generated output and evaluation, and occupation-specific oversight requirements can offset conditional displacement. As $e_j$ approaches zero, capability also operates on a vanishing execution component.

The observed outcome is an occupation-level employment stock rather than labor required per unit of fixed output. Its coefficient combines substitution, scale, labor supply, worker composition, and changes in the occupational task mix. The positive execution coefficient in the wage regression later in the paper makes this distinction consequential rather than semantic. I use equation~\eqref{eq:reg} to organize conditional associations. It does not identify $\sigma$, the oversight residual, or an AI treatment effect.

The empirical design has different leverage over the two implications. The additive gradients are estimated with reasonable precision. The interaction is harder to identify because $e_j a_j$ is collinear with the levels in a single occupation cross-section; both the continuous interaction and the family-adjusted super-additivity test are insignificant. Hiring or vacancy flows combined with predetermined adoption variation would be better suited to the multiplicative prediction in equation~\eqref{eq:Dstar}.

\subsection{Execution volume and immediate verification pressure}

The conditional labor-demand result holds occupational output fixed. Cheaper execution may also expand the number of candidate outputs that require evaluation. Let $G_j(A_t)$ denote candidate volume, with $G'_j(A_t)>0$, and let $V_{jt}$ be the predetermined stock of human verification capacity. A share $m_j\in[0,1]$ of checking is mechanically scalable; checking the remaining share consumes $\kappa_j>0$ units of human capacity per candidate. The fraction of candidates that can be checked is
\begin{equation}
\phi_{jt}=\min\left\{1,\;m_j+(1-m_j)\frac{V_{jt}}{\kappa_jG_j(A_t)}\right\}.
\label{eq:verificationpressure}
\end{equation}
The expression is deliberately spare. It describes a capacity constraint, not the error rate of unchecked output, the allocation of attention across candidates, or the institutional decision to record a detected problem. The first term, $m_j$, is the share checked mechanically. The second is the human-checked share: capacity $V_{jt}$ divided by the human load $\kappa_jG_j$, applied to the remaining $1-m_j$. As an illustration, with $m_j=0.2$ and capacity covering half the human load, the checked share is $0.2+0.8\times0.5=0.6$. If candidate volume doubles, it falls to $0.2+0.8\times0.25=0.4$; with $m_j=1$ it stays at one.

\begin{proposition}[Immediate verification pressure]
\label{prop:congestion}
Suppose the verification constraint in equation~\eqref{eq:verificationpressure} binds and $V_{jt}$ is predetermined at impact. A capability increase that raises candidate volume lowers the checked share, with the impact response scaled by the nonmechanical share $1-m_j$. A larger human verification requirement $\kappa_j$ lowers the checked share at every candidate volume and makes the constraint bind sooner.
\end{proposition}

\noindent\textit{Proof sketch.} In the binding region,
\[
\frac{\partial\phi_{jt}}{\partial A_t}
=-(1-m_j)\frac{V_{jt}}{\kappa_jG_j(A_t)^2}G'_j(A_t)<0.
\]
Mechanical coverage scales the exposed share directly. A larger $\kappa_j$ makes the constraint bind at a lower candidate volume and lowers the level of the checked share. At the fully mechanical limit $m_j=1$, $\phi_{jt}=1$ and the impact effect vanishes.

Proposition~\ref{prop:congestion} is conditional on a scale response, $G'_j(A_t)>0$. It does not imply that AI-generated outputs are intrinsically less reliable, nor does it sign an observed retraction or complaint hazard. Such outcomes combine the incidence of problems, the probability and speed of detection, and the probability of formal action after detection. An observed hazard does not by itself identify the verification-pressure mechanism in equation~\eqref{eq:verificationpressure}.

\subsection{Execution practice and the evaluator pipeline}

The execution-evaluation distinction also has a dynamic implication if evaluative judgment is produced by practice. Let $N_t=\nu(c^N_t)$ be the mass of novices, increasing in the entry-level execution resource $c^N_t$. Individual evaluation skill accumulates through the human execution work actually performed,
\begin{equation}
s^V_{t+1}=(1-\delta_V)s^V_t+\Phi(\lambda e_t s_t^*),
\label{eq:sV}
\end{equation}
and the stock of senior evaluators follows
\begin{equation}
S_{t+1}=(1-\delta_S)S_t+\mu N_{t-L},
\qquad L\ge1.
\label{eq:senior}
\end{equation}
The first equation is learning by doing: evaluation skill is produced by human execution practice, not by exposure to a cost share. The second is cohort accounting: today's senior stock depends on the entry cohort from $L$ periods earlier. Unlike the immediate pressure in Proposition~\ref{prop:congestion}, this supply response cannot appear until the affected cohort reaches seniority. \citet{afrouzi2026automation} develop the learning channel in general equilibrium, while \citet{dogonowski2026apprenticeship} and \citet{li2026generation} provide closely related recent formulations of apprenticeship and verification-capacity erosion. The reduced form here is retained to connect that dynamic mechanism to the occupation-level execution measure and to distinguish evaluator quantity from evaluation quality.

\begin{proposition}[Cohort-lag response]
\label{prop:pipeline}
Suppose entry-level execution falls at date zero and the novice mass declines from $N_0$ to $N_\infty<N_0$. Entry employment falls immediately. The path of the trained senior stock is unaffected for $L$ periods. Thereafter it converges geometrically toward $S_\infty=\mu N_\infty/\delta_S$ at rate $1-\delta_S$.
\end{proposition}

\noindent\textit{Proof sketch.} Before the affected cohort reaches seniority, equation~\eqref{eq:senior} continues to load on the pre-shock novice mass. After $L$ periods, the lower inflow enters the recursion, whose steady state and convergence rate follow directly. Appendix~\ref{app:formal} gives the full proof and the associated quality implications. With the depreciation rate $\delta_S=0.033$ of the illustrative scenario in Appendix~\ref{sec:estimation}, the gap to the new steady state halves about every $21$ years ($\ln 2/(-\ln 0.967)\approx20.7$). So half of the eventual decline arrives about two decades after the lag ends.

\begin{assumption}[Within-occupation seniority gradient]
\label{ass:seniority}
Entry-level work is more execution-heavy than senior work within an occupation, $e_{j,\mathrm{entry}}>e_{j,\mathrm{senior}}$.
\end{assumption}

Under Assumption~\ref{ass:seniority}, Proposition~\ref{prop:sorting} implies a larger near-term employment response for entry-level work at a given capability level. Equations~\eqref{eq:sV} and \eqref{eq:senior} add a distinct delayed prediction. Once post-shock cohorts reach seniority, both the number of evaluators and their average practice-produced skill may fall. Writing total evaluation capacity as $V_t^{\mathrm{tot}}=S_t\bar s_t^V$, the quantity effect operates through $S_t$ and the quality effect through $\bar s_t^V$.

The evidentiary status of the sequence differs at each stage. Proposition~\ref{prop:sorting} can be examined in occupation data, but the estimated execution gradient predates generative AI and largely lies between occupational families, so it does not establish AI-caused sorting. The near-term seniority implication can be examined only weakly in the CPS because the outcome misses hiring and occupational entry; that test is null. Proposition~\ref{prop:congestion} is not tested by the occupation data. The delayed senior-stock and evaluation-quality implications cannot be tested in a sample ending in 2025. A delayed fall in the \emph{number} of senior workers follows from any model in which entry feeds the senior stock. A fall in evaluative skill \emph{per senior worker}, and weaker evaluation outcomes conditional on evaluator numbers, additionally require the practice-to-skill channel. Appendix~\ref{app:formal} develops the derivations and comparison with alternative mechanisms; Appendix~\ref{sec:estimation} reports an illustrative timing exercise, not an estimate.

\section{Data and Measurement}
\label{sec:data}

\subsection{Task-level execution and AI-capable shares}
\label{sec:ajmeasure}

I score all $19{,}265$ O*NET task statements \citep{onet} under two separate rubrics. The execution rubric assigns high scores to producing, doing, or performing: operating, writing, computing, building, drafting, or carrying out a procedure. It assigns low scores to evaluating: assessing, verifying, auditing, approving, reviewing, or diagnosing. Mixed statements receive intermediate scores. The rubric explicitly separates this dimension from routineness: complex surgery is non-routine execution, while checking a document against a fixed standard is routine evaluation.

The capability rubric asks whether a frontier language model with typical tool access can perform the task as stated to a professional standard. It concerns technical feasibility, not whether deployment is lawful, licensed, accountable, or economical. Occupation titles are hidden in both scoring runs. For $k\in\{e,a\}$, the occupation score is
\begin{equation}
m_j^k=\frac{\sum_{t\in j}w_t m_t^k}{\sum_{t\in j}w_t},\qquad
w_t=\begin{cases}2&\text{core task},\\1&\text{supplemental task}.\end{cases}
\label{eq:measures}
\end{equation}
Here $m_t^e$ is the execution score and $m_t^a$ is the capability score. Equal, importance-based, and frequency-based task weights leave both the occupation scores and the six-year employment coefficient essentially unchanged. Appendix~\ref{app:measurement} reports those figures with the batch-order and sample-flow details, and Appendix~\ref{app:rubric} gives the prompts and worked examples. Both shares are estimated rather than observed, so every occupation characteristic used below is a generated regressor. The standard errors reported throughout do not propagate scoring uncertainty; Appendix~\ref{app:measurement} quantifies its size.

The baseline task universe is O*NET version 26.1 from November 2021, before the diffusion of generative AI. Rebuilding the execution share on version 30.3 from May 2026 produces an occupation-level correlation of $0.996$. The two capability versions correlate $0.998$. So task-description revisions after 2021 do not account for the occupational rankings.

The capability judgment itself has a later vintage. The scoring model reflects the frontier in 2026, even though it evaluates 2021 task descriptions. So the \emph{AI-capable share} $a_j$ carries a 2026 vintage and is not an ex ante exposure to the 2022 shock. I use two external anchors. The first is the early-2023 GPT-4 exposure measure of \citet{eloundou2024gpts}, which I call \emph{Eloundou exposure}. It correlates $0.80$ with the AI-capable share among cognitive occupations and serves as the vintage-valid alternative. The second is the Anthropic Economic Index's observed-usage exposure, which correlates $0.60$ across all matched occupations and $0.48$ within the cognitive sample. Its conversations begin in December 2024, so usage is contemporaneous with the 2022--2025 employment outcome and is not treated as predetermined adoption. \citet{massenkoff2026labor} build an occupational exposure measure from the same conversation data, weighting theoretical capability by observed and automation-oriented usage into a single index. The measures are related but not substitutable. One ranking that combines capability and usage is the efficient choice for forecasting displacement risk. The question here needs them separate. The joint cells and the interaction test exist only when execution content, capability, and usage are three distinct occupation characteristics scored on a common task universe.

\subsection{Employment, wages, seniority, and samples}

Employment and mean wages come from the national cross-industry Occupational Employment and Wage Statistics estimates \citep{bls_oews}. The six-year outcome is
\[
D_j=\ln(\mathrm{emp}_{j,2025})-\ln(\mathrm{emp}_{j,2019}).
\]
Both endpoints use the 2018 SOC system. OEWS pools three years of survey panels, so adjacent annual changes are smoothed and are not independent. The long-run extension uses May 2012 and May 2016 files under SOC-2010, stable crosswalk components through SOC-2018, and BLS research estimates that hold the post-2021 MB3 estimator constant where available.

The principal sample contains the ten white-collar SOC major groups, which I call the \emph{cognitive occupations}. They are management; business and financial operations; computer and mathematical; architecture and engineering; life, physical, and social science; legal; education; arts and media; healthcare practitioners; and office and administrative support. I refer to the SOC major groups as \emph{occupational families}. The task data include $400$ cognitive occupations. The 2019--2025 employment regressions contain $337$ cognitive occupations, or $284$ after routine intensity and interest-rate exposure are required. The post-2022 usage sample contains $364$ occupations, and the harmonized long-run panel contains $323$ cognitive occupation groups. Appendix~\ref{app:measurement} reports the complete sample flow.

Table~\ref{tab:desc} reports descriptive statistics for the estimation samples. The seniority analysis uses IPUMS CPS ASEC data from 2019 through 2025 \citep{ipumscps}. It compares labor-force participants aged 22--30 with those aged 35--64 and maps the task measures to Census 2010 occupations. The outcome is employment conditional on labor-force participation. It misses labor-force exit and workers who never enter an occupation, which are important limits for interpreting the null.

\begin{table}[t]
\centering
\caption{Descriptive statistics}
\label{tab:desc}
\small
\begin{tabular}{lcccc}
\toprule
 & \multicolumn{2}{c}{All ($n=798$)} & \multicolumn{2}{c}{Cognitive ($n=400$)} \\
\cmidrule(lr){2-3}\cmidrule(lr){4-5}
 & Mean & SD & Mean & SD \\
\midrule
Execution share $e_j$    & 0.68 & 0.11 & 0.62 & 0.09 \\
AI-capable share $a_j$   & 0.36 & 0.19 & 0.49 & 0.14 \\
\bottomrule
\end{tabular}
\par\vspace{4pt}
\footnotesize\raggedright Note: the execution share is the language-model measure with full coverage. The cognitive subset is the ten white-collar SOC major groups (11, 13, 15, 17, 19, 23, 25, 27, 29, 43).
\end{table}

\subsection{Separability, routine intensity, and validation}

Within cognitive occupations, the execution share and the AI-capable share correlate $-0.35$ in the index data and $-0.34$ in the OEWS match. Execution share retains about $94$ percent of its standard deviation after capability is partialled out, and all four execution-capability quadrants are populated (Figure~\ref{fig:sep}). Financial examiners, underwriters, credit analysts, and personal financial advisers are relatively evaluation-heavy at high capability; programmers and mathematicians are more execution-heavy at comparable capability.

\begin{figure}[t]
\centering
\includegraphics[width=\textwidth]{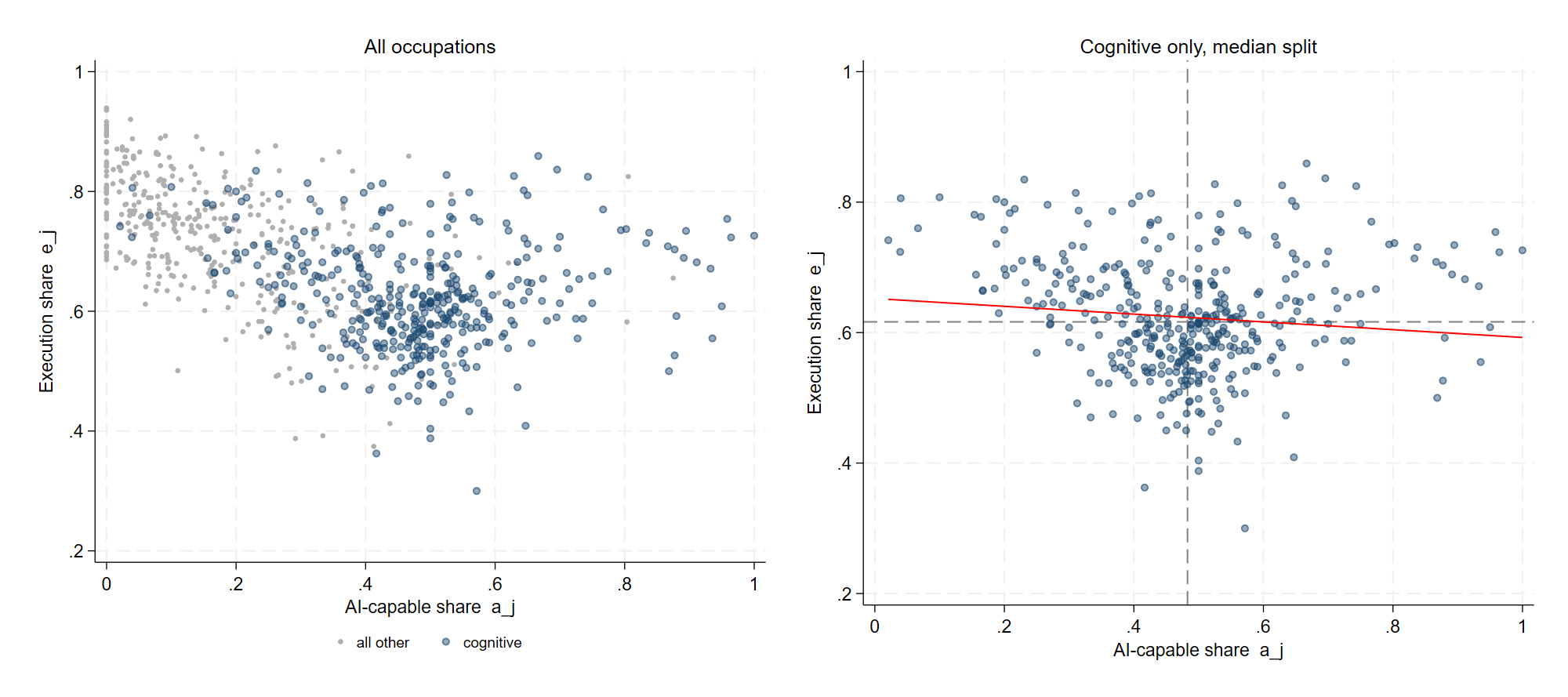}
\caption{Execution share $e_j$ and AI-capable share $a_j$. Left: all occupations, with cognitive occupations highlighted. Right: cognitive occupations with median splits and a linear fit. Within the cognitive sample, the measures correlate $-0.35$ and populate all four quadrants.}
\label{fig:sep}
\end{figure}

Execution share is positively but imperfectly related to routine-task intensity. The baseline routine index is the one carried through the estimation. The reconstruction follows \citet{acemoglu2011skills} exactly: routine cognitive and routine manual O*NET descriptors minus non-routine analytical, interpersonal, and manual descriptors, each standardized across six-digit occupations. The two correlate $0.89$. Replacing the baseline with the reconstructed measure weakens the execution coefficient from $-0.70$ to $-0.45$ but leaves it significant ($p=0.023$). The weakening fits the reconstruction's non-routine components, which include evaluation-adjacent activities and so absorb part of the execution dimension. Within cognitive occupations, execution share and routine intensity correlate $0.43$, so routine content explains about one-fifth of execution-share variation. AI capability is essentially uncorrelated with routine intensity ($r=-0.07$).

The distinction is conceptual as well as statistical. Routine intensity classifies the procedure: whether it is repetitive, codifiable, and rule-bound. Execution share classifies the function: whether the worker produces an output or judges it. Oral surgeons, floral designers, mathematicians, and physicists are relatively non-routine executors. Proofreaders, sports officials, accountants, and underwriters are relatively routine evaluators. Figure~\ref{fig:ejrti} plots the two measures against one another. Appendix~\ref{app:measurement} lists these off-diagonal occupations and reports window-specific routine controls.

\begin{figure}[t]
\centering
\includegraphics[width=0.82\textwidth]{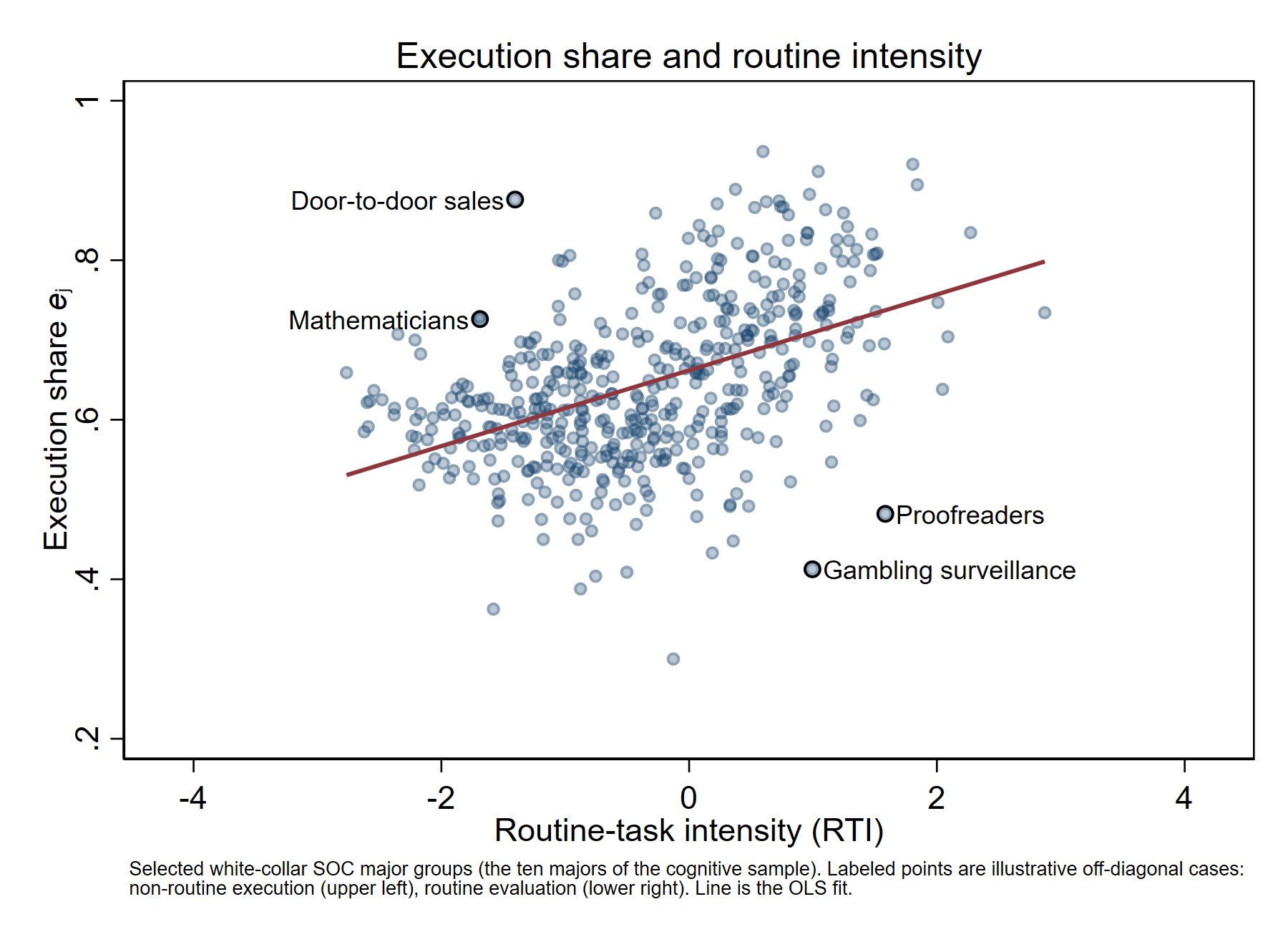}
\caption{Execution share and routine-task intensity in cognitive occupations. The relationship is positive but imperfect. Labeled observations illustrate non-routine execution (upper left) and routine evaluation (lower right).}
\label{fig:ejrti}
\end{figure}

A task-level audit provides a separate reproducibility check. Two independent model coders, blind to occupations, production scores, and each other, scored a stratified sample of $500$ statements under different framings. Their scores correlate $0.991$. Their consensus correlates $0.924$ with the production score, with ICC$(A,1)=0.908$ and a mean absolute error of $10.1$ points on the $0$--$100$ scale. There is no systematic failure class across drafting, reviewing, diagnostic, and supervisory tasks. Because all coders are language models, this is evidence of cross-model reproducibility rather than human construct validation.

A second audit provides the direct human benchmark. Using seed $20260727$, I drew a simple random sample without replacement of $200$ statements from the $19{,}265$-task frame. The draw spans $180$ occupations. Four human coders independently scored every statement under the same execution--evaluation definition while blind to occupation titles, employment outcomes, production scores, and one another's answers. The first two coders received English materials; the later two received bilingual materials in which the English statement controlled any discrepancy. I use the unadjudicated mean of all four scores as the benchmark, preserving disagreement rather than resolving it in favor of the production model.

Human judgment is informative but not homogeneous. Coder means range from $55.5$ to $74.3$, the mean pairwise correlation is $0.522$, and the mean pairwise absolute difference is $23.8$ points. Absolute agreement for one coder is moderate, ICC$(A,1)=0.482$, while the reliability of the four-coder average is $0.788$. Against that average, the production score has Pearson correlation $0.806$ (bootstrap $95$ percent interval $[0.745,0.857]$), Spearman correlation $0.783$, mean absolute difference $12.3$ points, and mean bias of $-1.4$ points. Three-way execution/mixed/evaluation agreement is $75.5$ percent, with $\kappa=0.562$. Doubao and Gemini yield similar agreement, so the result is not specific to the production model. The mean of the first two coders and the mean of the later bilingual coders correlate $0.659$, indicating protocol robustness alongside meaningful scale-use differences.

Two comparisons put those figures in context, because a correlation of $0.806$ against a noisy benchmark understates agreement. The first is the benchmark's own reliability, which bounds what any scorer can achieve against it. With ICC$(A,4)=0.788$, even a perfectly valid measure recorded without error would correlate only about $0.89$ with this particular four-coder mean. The production score's $0.806$ sits close to that ceiling, and the implied correlation with an error-free consensus is about $0.91$. The second comparison is direct. I hold out each coder in turn and predict the mean of the remaining three. The held-out human correlates $0.674$, $0.544$, $0.646$, and $0.660$ with that target. The production score correlates $0.750$, $0.808$, $0.773$, and $0.793$ with the same four targets, higher in every case and by $0.15$ on average. Against the consensus of human coders, in other words, the score is a better predictor than an individual human coder is.

That is not a claim to ground truth, and the dispersion across coders is the reason. Task-level execution content admits genuine interpretive variance. Reasonable readers of the same statement differ by more than twenty points on average, and no adjudication in this design resolves which of them is right. What the audit establishes is narrower and sufficient for the use made of the measure here. The score lies inside the distribution of human judgment and near its center: the production mean is $1.4$ points from the four-coder mean, while the coders themselves span $18.8$ points. So the occupation-level aggregates that follow do not track an artifact of one scoring model. Two limits remain. First, the audit is task-level. Only $16$ of the $180$ sampled occupations contribute more than one statement. So the greater reliability of the occupation-level average, which pools roughly two dozen tasks, is inferred from that structure rather than measured. Second, agreement is assessed against a mean of four coders, not against an external criterion; no such criterion for this construct exists.

\begin{table}[H]
\centering
\caption{Human and model validation of the task-level execution score}
\label{tab:humanaudit}
\small
\begin{tabular}{lccc}
\toprule
 & Estimate & & \\
\midrule
\multicolumn{4}{l}{\textit{Panel A: Human-coder agreement}} \\
Coder means (1--4) & \multicolumn{3}{c}{$68.8,\ 55.5,\ 66.3,\ 74.3$} \\
Mean pairwise Pearson correlation & $0.522$ & & \\
Mean pairwise Spearman correlation & $0.516$ & & \\
Mean pairwise absolute difference & $23.8$ & & \\
Mean three-way agreement & $0.581$ & & \\
ICC$(A,1)$ & $0.482$ & & \\
ICC$(A,4)$ & $0.788$ & & \\
\addlinespace
\multicolumn{4}{l}{\textit{Panel B: Score compared with the four-human mean}} \\
 & Production & Doubao & Gemini \\
Pearson correlation & $0.806$ & $0.796$ & $0.810$ \\
Spearman correlation & $0.783$ & $0.802$ & $0.794$ \\
Mean absolute difference & $12.3$ & $12.6$ & $11.7$ \\
Mean score minus human mean & $-1.4$ & $-3.0$ & $-1.1$ \\
Three-way agreement & $0.755$ & $0.740$ & $0.770$ \\
Three-way Cohen's $\kappa$ & $0.562$ & $0.538$ & $0.576$ \\
\bottomrule
\end{tabular}
\par\vspace{4pt}
\footnotesize\raggedright Note: Panel A summarizes six human-coder pairs. ICC is the two-way absolute-agreement coefficient; ICC$(A,4)$ is the reliability of the four-coder mean. Panel B compares each model score with the unadjudicated mean of four humans. Three-way categories are evaluation ($0$--$35$), mixed ($36$--$64$), and execution ($65$--$100$). All statistics are unweighted over the fixed-seed simple random sample of $200$ task statements. The audit provides a direct benchmark for task-level scoring; it does not validate the occupation-level aggregation or the employment design.
\end{table}

\subsection{What the task-level measure adds beyond O*NET evaluation activities}
\label{sec:onetbenchmark}

O*NET already provides occupation-level ratings for evaluating information against standards, judging qualities, monitoring, inspecting, guiding subordinates, coaching, and consultation. These independently collected activities are the natural benchmark. I regress the execution share on an evaluation-activity composite, a rebuilt routine index, and a manual-activity composite. The coefficient on the evaluation composite is $-0.058$ (robust $t=-10.5$) across $759$ matched occupations and $-0.045$ ($t=-6.4$) within cognitive occupations. The corresponding $R^2$ values are $0.45$ and $0.33$. This strong negative relation is convergent validity, not evidence that the execution score recovers an entirely new latent trait. Table~\ref{tab:onetcompare} sets out what the task-level measure adds beyond the O*NET activities and what it does not.

\begin{table}[t]
\centering
\caption{What the task-level execution share adds to O*NET evaluation activities}
\label{tab:onetcompare}
\small
\begin{tabular}{@{}>{\raggedright\arraybackslash}p{0.23\linewidth}>{\raggedright\arraybackslash}p{0.35\linewidth}>{\raggedright\arraybackslash}p{0.35\linewidth}@{}}
\toprule
 & Task-level execution share & O*NET evaluation-activity composite \\
\midrule
Unit of observation & Individual O*NET task statements, aggregated to occupations & Occupation-level Work Activity ratings \\
Alignment with AI capability & Same statements, weights, and occupational aggregation as the capability score & Separate descriptors and rating process \\
Joint execution--capability mapping & Available at the task and occupation level & Not available from the composite alone \\
Portability & Fixed rubric can be applied to other task corpora and O*NET vintages & Tied to the selected O*NET activity items \\
Validation status & Human-benchmarked on $200$ blinded tasks; production score $r=0.806$ with four-coder mean & Independently collected human O*NET ratings \\
Incremental 2019--2025 fit & $\Delta R^2=0.006$ over the composite alone; $0.011$ in the full controlled comparison & Strong independent predictor; $\Delta R^2=0.058$ in the full controlled comparison \\
\bottomrule
\end{tabular}
\par\vspace{4pt}
\footnotesize\raggedright Note: the final row reports the comparisons described in the text. The execution share's main incremental value is task-level alignment with AI capability and portability, not superior predictive fit relative to the existing O*NET activity composite.
\end{table}

The employment comparison reinforces that interpretation. In the common cognitive sample ($n=335$), execution share and the evaluation composite correlate $-0.35$. Employment change correlates $-0.22$ with execution share and $+0.41$ with the composite. The composite's coefficient is nearly the same when entered alone and conditional on the execution share ($+0.099$ and $+0.092$ per standard deviation, $t=7.8$). So its positive sign is a genuine association, not an artifact of suppression by the execution share. Conditional on the composite alone, the execution coefficient falls to $-0.020$ per standard deviation ($t=-1.5$) and adds $\Delta R^2=0.006$.

In the full controlled employment specification ($n=284$), execution share retains a coefficient of $-0.376$ ($t=-1.95$) beside the composite's $+0.104$ ($t=4.9$). Its incremental $R^2$ is $0.011$, compared with $0.058$ for the composite. So evaluation-heavy occupations grew over 2019--2025, and the task-level execution share adds a modest, control-dependent increment. Its distinctive contribution is the common task-level architecture. It can be paired directly with capability scores, split into joint cells, rebuilt on other vintages, and carried to task databases without O*NET Work Activity ratings.

\section{Employment Facts}
\label{sec:results}

\subsection{The long-run gradients, 2012--2025}
\label{sec:timing}

The chronology is the central empirical result. A 2019--2022 ``pre-period'' spans the pandemic and is not a clean baseline for generative AI. So I harmonize OEWS employment back to 2012, retaining only crosswalk components with stable reporting membership. The balanced panel contains $323$ cognitive occupation groups. Of the $669$ groups in the full panel, $661$ are one-to-one code matches and eight are consistent two-to-one mergers from SOC-2010 to SOC-2018. Excluding the mergers leaves the results essentially unchanged.

\begin{table}[t]
\centering
\caption{Pre-pandemic extension: annualized window slopes on the harmonized 2012--2025 panel}
\label{tab:pretrend}
\small
\begin{tabular}{lcccc}
\toprule
Annualized $\Delta\ln$ emp.\ on: & 2012--16 & 2016--19 & 2019--22 & 2022--25 \\
\midrule
\multicolumn{5}{l}{\emph{Panel A. Estimator-consistent (MB3 research estimates for 2016 and 2019)}} \\
Execution share (univariate) &  & $-0.139^{***}$ & $-0.095^{***}$ & $-0.136^{***}$ \\
                             &  & $(0.029)$      & $(0.030)$      & $(0.034)$ \\
Execution share (joint)      &  & $-0.151^{***}$ & $-0.085^{**}$  & $-0.116^{***}$ \\
                             &  & $(0.032)$      & $(0.033)$      & $(0.037)$ \\
Eloundou exposure (joint)  &  & $-0.056^{***}$ & $-0.035^{**}$  & $-0.089^{***}$ \\
                             &  & $(0.020)$      & $(0.018)$      & $(0.017)$ \\
\addlinespace
\multicolumn{5}{l}{\emph{Panel B. Published vintages (old method through 2019; MB3 from 2022)}} \\
Execution share (univariate) & $-0.088^{***}$ & $-0.136^{***}$ & $-0.103^{***}$ & $-0.136^{***}$ \\
                             & $(0.025)$      & $(0.026)$      & $(0.033)$      & $(0.034)$ \\
Execution share (joint)      & $-0.105^{***}$ & $-0.125^{***}$ & $-0.100^{***}$ & $-0.116^{***}$ \\
                             & $(0.027)$      & $(0.025)$      & $(0.036)$      & $(0.037)$ \\
Eloundou exposure (joint)  & $-0.045^{***}$ & $-0.072^{***}$ & $-0.054^{***}$ & $-0.089^{***}$ \\
                             & $(0.015)$      & $(0.018)$      & $(0.018)$      & $(0.017)$ \\
\bottomrule
\end{tabular}
\par\vspace{4pt}
\footnotesize\raggedright Note: OLS on $323$ cognitive occupation groups observed in all five OEWS vintages. The stable crosswalk contains $661$ one-to-one code matches and eight consistent two-to-one 2010-to-2018 mergers; components whose membership changes are excluded. The restriction drops $30$ cognitive occupations ($4.4$ percent of 2019 cognitive employment), whose mean execution share and Eloundou exposure are close to those of the retained sample.
\par\smallskip Outcomes are annualized log employment changes; heteroskedasticity-robust standard errors are in parentheses, and joint rows include routine intensity. Panel A uses BLS MB3 research estimates for May 2016 and May 2019, making the 2016--2019, 2019--2022, and 2022--2025 windows estimator-consistent. Equality of the execution slopes is not rejected in the univariate ($p=0.41$) or joint ($p=0.20$) specification.
\par\smallskip Panel B uses published vintages, so 2019--2022 spans the estimator change; its four-window execution equality test gives $p=0.21$, with window-minus-2012--16 differences carrying $95$ percent confidence intervals $[-0.096,\ 0.000]$, $[-0.086,\ 0.056]$, and $[-0.113,\ 0.018]$ and a minimum detectable slope difference of about $0.089$ per year. On Panel A, equality of the Eloundou exposure slopes is rejected (joint $p=0.020$; univariate $p=0.013$), driven by 2022--2025 relative to 2019--2022. The execution result is nonrejection of equality, not proof of constancy. $^{+}\,p<0.10$, $^{**}\,p<0.05$, $^{***}\,p<0.01$.
\end{table}

Table~\ref{tab:pretrend} reports annualized window slopes on this constant sample. The execution gradient is negative in every window. In the published-vintage panel, the univariate annualized execution slopes are $-0.088$ in 2012--2016, $-0.136$ in 2016--2019, $-0.103$ in 2019--2022, and $-0.136$ in 2022--2025. Equality across the four windows is not rejected ($p=0.21$). Because the May 2021 release introduced the MB3 estimator, Panel A holds the estimator constant using BLS research estimates for 2016 and 2019. The corresponding slopes are $-0.139$, $-0.095$, and $-0.136$, and equality is again not rejected ($p=0.41$ univariately; $p=0.20$ jointly). These are nonrejections rather than proof of constancy, and their power is limited. The minimum detectable slope difference is about $0.089$ per year, against slopes of roughly $0.1$, so the design can reject only large changes in slope. The execution slopes are also less precise than the capability slopes reported below (standard errors of $0.029$ to $0.034$ against $0.017$ to $0.020$). What the two clean pre-pandemic windows do establish, independently of any equality test, is that the execution gradient predates generative AI.

Eloundou exposure is also significantly negative before 2022. Its estimator-consistent joint slopes are $-0.056$ in 2016--2019, $-0.035$ in 2019--2022, and $-0.089$ in 2022--2025. Unlike the execution slopes, equality is rejected ($p=0.020$). The 2022--2025 slope is $0.054$ per year more negative than the pandemic-window slope (s.e. $0.021$, $p=0.009$), although it is not significantly different from 2016--2019. The supported statement is asymmetric: the execution gradient shows no detected break, while the capability gradient is longstanding but steepens after 2022 relative to the pandemic window. The design dates that change; it does not identify AI as its cause.

The panel is an intensive-margin design: each observation is an occupation present at both endpoints, so the slopes are silent on occupations entering or leaving the measured universe. That margin is almost entirely entry. By 2022--2024, $51$ cognitive occupations appear with no 2019 baseline, against a single exit. The entrants are statistically indistinguishable from the balanced set on both characteristics (AI-capable share $0.49$ against $0.49$; execution share $0.61$ against $0.63$). So the estimates are not attenuated by the hardest-hit occupations falling below the reporting threshold. Restoring the entrants moves the univariate 2022--2024 slope on the AI-capable share from $-0.14$ to $-0.11$ (Appendix Table~\ref{tab:dissoc}). Much of the entry is a taxonomy change rather than occupation birth. Data scientists, database architects, and the software-developer split are 2018-SOC codes that OEWS phased into separate estimation after the May 2019 file. Appendix~\ref{app:employmentrobust} reports the three-window decomposition for the AI-capable share, which shows a suggestive window pattern but no rejected break. It also reports the broader sample that adds in-person service occupations, with the associated sample-definition sensitivity. Those exercises reinforce two cautions: significance in one smoothed window and not another is not a test of slope equality, and any 2019--2022 comparison is also a pandemic comparison.

\subsection{The 2019--2025 cross-section and occupational families}
\label{sec:families}

The six-year cross-section summarizes which occupations grew over the full period. Table~\ref{tab:reg} estimates equation~\eqref{eq:reg} without the interaction, which is imprecise. Within cognitive occupations, both execution share and capability are negatively associated with employment change. Controlling for routine intensity and interest-rate exposure yields coefficients of $-0.701$ for execution and $-0.468$ for capability. The interest-rate control is included because \citet{iscenko2026ladder} argues that AI-exposed occupations are concentrated in rate-sensitive sectors. It enters insignificantly ($-0.025$, s.e.\ $0.061$), so that channel does not drive the occupation-level result. Column 1 reports the pooled all-occupation estimate for reference. There capability carries the opposite sign, a reminder that the cognitive restriction is part of the specification rather than a neutral preliminary.

\begin{table}[H]
\centering
\caption{Execution share and the 2019--2025 employment change, language-model measure}
\label{tab:reg}
\small
\begin{tabular}{lccc}
\toprule
 & (1) & (2) & (3) \\
 & $\Delta\ln$ emp & $\Delta\ln$ emp & $\Delta\ln$ emp \\
 & all & cognitive & cognitive \\
 & baseline & baseline & +controls \\
\midrule
AI-capable share $a_j$ & $\phantom{-}0.017$ & $-0.499^{***}$ & $-0.468^{***}$ \\
                       & $(0.081)$ & $(0.110)$ & $(0.124)$ \\
Execution share $e_j$  & $-0.692^{***}$ & $-0.841^{***}$ & $-0.701^{***}$ \\
                       & $(0.139)$ & $(0.171)$ & $(0.192)$ \\
Routine intensity      &  &  & $-0.051^{***}$ \\
                       &  &  & $(0.016)$ \\
Rate exposure          &  &  & $-0.025$ \\
                       &  &  & $(0.061)$ \\
\midrule
Observations           & 708 & 337 & 284 \\
\bottomrule
\end{tabular}
\par\vspace{4pt}
\footnotesize\raggedright Note: OLS, robust standard errors in parentheses. The execution share is the language-model measure with full coverage. Outcome is the log employment change from May 2019 to May 2025. Routine intensity is the standardized routine-minus-abstract O*NET index; rate exposure is the share of private employment in interest-rate-sensitive sectors. Column-3 variance inflation factors are all below 1.5. The wage-change counterpart of column 3 is Table~\ref{tab:wage}. $^{+}\,p<0.10$, $^{**}\,p<0.05$, $^{***}\,p<0.01$.
\end{table}

The key limitation is that these gradients lie mainly between occupational families. Adding family fixed effects reduces the execution coefficient from about $-0.70$ to $-0.14$ and removes its significance; the capability coefficient also becomes insignificant. So the main six-year association compares families such as office support, management, education, health, and computing, not execution-heavy and evaluation-heavy occupations within the same family. Permanent differences in education, regulation, industry composition, remote-work feasibility, offshorability, and secular demand can generate that pattern.

The controlled execution coefficient still survives two narrower comparisons. Replacing the baseline routine index with the rebuilt \citet{acemoglu2011skills} measure gives $-0.45$ ($p=0.023$). Replacing the AI-capable share with Eloundou exposure leaves execution at $-0.738$ and gives capability at $-0.397$ ($t=-4.3$). This uses the $275$ of the $284$ controlled-specification occupations that carry an Eloundou score. Eloundou exposure is earlier but not predetermined either: released in 2023, it postdates the start of the 2019--2022 window in which its slope is reported. When both capability measures enter, the AI-capable share becomes insignificant while Eloundou exposure stays negative. So the full-period capability gradient belongs to their common component, not to information unique to the later scoring model.

Because the variation is concentrated across roughly ten occupational families, I also use family-level inference. A wild-cluster bootstrap clusters on the ten families, imposes the null, and enumerates all Rademacher weights. It gives $p=0.029$ for execution and $p=0.041$ for the AI-capable share in the controlled specification. With Eloundou exposure, the corresponding values are $0.006$ and $0.012$. Employment weighting, precision weighting, trimming small occupations, and leave-one-family-out estimates preserve the sign of the execution coefficient. These robustness results are reported in Table~\ref{tab:weights} and Appendix~\ref{app:employmentrobust}. They show that the between-family pattern is not generated by one small or unusual occupation, but they do not substitute for the null within-family fixed-effects result. A common-sample check separates the role of the controls from the loss of observations across columns of Table~\ref{tab:reg}. On the $284$ occupations of the controlled specification, the uncontrolled coefficients are $-0.94$ for execution and $-0.51$ for capability. So the movement to $-0.70$ and $-0.47$ reflects the controls rather than the change in sample.

The negative execution coefficient is equivalently a positive association for evaluation-heavy occupational families, because the two shares sum to one. Rising demand for health, education, management, and regulated professional services is an alternative description of the same six-year fact. The O*NET evaluation-activity composite points in the same direction and has stronger predictive content, as Section~\ref{sec:onetbenchmark} documents.

\subsection{Within occupations: execution share and industry AI adoption}
\label{sec:adoption}

The family pattern can be tested inside occupations. OEWS also publishes national employment for each occupation within each industry, so the same occupation can be compared across industries. That comparison removes every occupation-level trend, including the family differences that absorb most of the six-year gradient once family fixed effects enter. The question becomes narrower: did execution-heavy occupations shrink more in industries that already used AI before generative AI? I call this difference the \emph{adoption gradient}. I estimate
\begin{equation}
\Delta\ln \mathrm{emp}_{ji}=\alpha_j+\gamma_i+b\,(e_j\times \mathrm{AI}_i)+c\,(a_j\times \mathrm{AI}_i)+u_{ji},
\label{eq:adoption}
\end{equation}
where $\mathrm{AI}_i$ is industry $i$'s AI adoption measured before 2022. The occupation effect $\alpha_j$ absorbs everything an occupation carries across industries: its family, education, licensing and secular demand. The industry effect $\gamma_i$ absorbs the industry's overall growth. The coefficient $b$ is the adoption gradient, the extra annual employment change of a one-standard-deviation more execution-heavy occupation in an industry one standard deviation higher in prior AI use. If adoption raises the effective capability that execution work faces, equation~\eqref{eq:Dstar} predicts $b<0$.

Two predetermined measures supply $\mathrm{AI}_i$. The first is sector AI use from the 2018 Annual Business Survey, which asked firms whether they used AI in producing goods or services in reference year 2017 \citep{census_abs2018,zolas2020advanced}. It is published for $\iaNSectors$ private sectors. The share of firms using AI runs from $\iaFirmMin$ percent in agriculture to $\iaFirmMax$ percent in information. This share predicts later adoption: across sectors it correlates $\iaRelFirm$ with the share of firms using AI in 2025--26 in the Census Business Trends and Outlook Survey \citep{census_btos}. The share of employment in AI-using firms correlates only $\iaRelEmp$, because it tracks a few large firms in retail, transportation and manufacturing; it correlates $\iaCorrEmpRobot$ with robot use. I report both versions. The second measure is finer. For $\iaFineUnits$ industries at the three- and four-digit level, it is the share of May 2016 employment in AI-specialist occupations. These are computer and information research scientists, mathematicians, statisticians and the residual mathematical science occupations, where data scientists were coded before the 2018 SOC. It follows the staffing approach of \citet{babina2024ai}. Aggregated, it behaves like adoption: it correlates $\iaFineRelABS$ with the ABS firm share across sectors and $\iaFineRelSub$ with BTOS adoption across $\iaFineRelSubN$ three-digit industries. It is also thin. The share is zero in $\iaFineZero$ of the $\iaFineUnits$ industries because OEWS suppresses small cells, and within sectors it correlates only $\iaFineRelWithin$ with later adoption. So I use its percentile rank.

The windows follow Table~\ref{tab:pretrend}, plus the six-year window, all annualized. The 2016--2019 window uses the MB3 research estimates at both ends and serves as the pre-period placebo. Industry codes are harmonized across the NAICS 2022 revision, which OEWS adopted in May 2022, and the AI-specialist occupations are dropped from the outcome cells of the industry design. Adoption varies across only $\iaNSectors$ sectors, so the $p$-values come from a wild cluster bootstrap that resamples sectors. Appendix Table~\ref{tab:adoption} reports every estimate.

In the sector design, the adoption gradient is significant for one predetermined measure in one window. With the share-of-firms measure, $b=\iaFirmBPost$ in 2022--2025 ($p=\iaFirmPPost$). In a sector one standard deviation higher in 2017 AI use, a one-standard-deviation more execution-heavy occupation grew $\iaFirmAbsPost$ percent per year more slowly than the same occupation elsewhere. That gradient is $\iaFirmPostShare$ percent of the national execution slope in the same window, $\iaNatBPost$ percent per year per standard deviation. It is not new. The 2016--2019 estimate is similar in size, $\iaFirmBPre$ ($p=\iaFirmPPre$), and the post-minus-pre difference is $\iaFirmBDiff$ ($p=\iaFirmPDiff$). The pandemic window reverses the sign ($\iaFirmBPan$, $p=\iaFirmPPan$), and over 2019--2025 the gradient is $\iaFirmBSix$ ($p=\iaFirmPSix$). The employment-share measure gives gradients near zero in every window. It was the planned primary measure. I moved the share-of-firms measure to the lead after the relevance check above. That check uses no employment data, but the employment results of both measures had been seen by then.

At the sector level, AI adoption cannot be separated from digital adoption in general. Across sectors, the AI firm share correlates $\iaCorrCloud$ with cloud use and $\iaCorrSoft$ with specialized software. When the three measures enter together, the software gradient is $\iaHrSoftBPost$ in 2022--2025 ($p=\iaHrSoftPPost$) and the AI gradient turns positive ($\iaHrAIBPost$, $p=\iaHrAIPPost$).

At the industry level there is no adoption gradient. With occupation and industry effects, the 2022--2025 gradient on AI-specialist staffing is $\iaFineOccBPost$ (s.e.\ $\iaFineOccSPost$). Its 95 percent confidence interval excludes negative gradients larger than $\iaFineLowerShare$ percent of the national execution slope. Occupation-by-sector effects, which leave only variation across industries within a sector, give $\iaFineSecBPost$ ($p=\iaFineSecPPost$). The one significant single window within sectors is 2019--2022, and its sign is positive ($\iaFineSecBPan$, $p=\iaFineSecPPan$), the opposite of an adoption gradient. One contrast is significant: within sectors, the gradient moves from $\iaFineSecBPre$ in 2016--2019 to $\iaFineSecBPost$ in 2022--2025, a change of $\iaFineSecBDiff$ ($p=\iaFineSecPDiff$). I treat it as fragile for three reasons. Neither window is significant alone. The contrast rests on the variation with the weakest link to later adoption. It also disappears across sectors ($\iaFineOccBDiff$, $p=\iaFineOccPDiff$). Realized adoption gives the same answer. As a post-outcome ceiling check, I replace the predetermined measure with 2025--26 BTOS adoption. Each of $\iaFineBtosUnits$ industries takes the value of its three-digit parent, or of its sector for retail and information. The 2022--2025 gradient is then $\iaFineBtosOccBPost$ ($p=\iaFineBtosOccPPost$) across industries and $\iaFineBtosSecBPost$ within sectors.

Three conclusions follow. First, inside occupations, the evidence for an adoption gradient is weak. The one significant sector estimate, for 2022--2025, has an insignificant counterpart of similar size in 2016--2019 and turns positive once software adoption enters. Second, whatever pattern exists runs between sectors: inside a sector, neither 2016 AI staffing nor realized 2025--26 adoption separates execution-heavy occupations in 2022--2025. Third, the design does not attribute any part of the execution gradient to AI adoption.

\subsection{Post-2022 concentration in observed usage and execution}
\label{sec:usage}

Observed usage helps locate the recent decline but does not provide predetermined adoption variation. The Anthropic measure is based on conversations from late 2024 onward, inside the 2022--2025 outcome window, and usage on one platform is not employer adoption. Employment changes can influence usage, and common demand shocks can influence both.

Table~\ref{tab:usage} reports the estimates. Column 1 compares usage with the AI-capable share. A one-standard-deviation increase in usage predicts a $2.6$ percent larger decline ($t=-3.1$), while the AI-capable share is insignificant conditional on usage. The comparison reverses with Eloundou exposure: usage weakens to $-1.8$ percent ($p=0.089$), and Eloundou exposure is $-2.6$ percent ($t=-3.1$). Usage and Eloundou exposure point to nearly the same occupations and cannot be separated as correlates.

Column 2 tests the continuous execution-by-usage interaction. Execution and usage each predict contraction conditional on the other, but the interaction is negative and insignificant ($-0.011$, $p=0.34$). Columns 3--5 report median-split cells. Relative to occupations low in both dimensions, the two off-diagonal cells decline by an insignificant $3$ to $3.5$ percent, while the high-high cell declines by $11.5$ percent ($t=-5.3$). In raw means, it is the only cell with an absolute decline. Adding capability, routine intensity, and family fixed effects reduces the high-high coefficient to $-0.056$ but leaves it significant ($p=0.009$). What survives the family fixed effects is that cell, not an execution gradient within families. On the continuous margin in this window, usage stays negative within families ($-0.024$, $t=-2.5$), while the execution main effect does not ($-0.009$, insignificant).

The cell pattern does not establish super-additivity. The formal double difference $\beta_{HH}-\beta_{HL}-\beta_{LH}$ is $-0.047$ ($p=0.14$) in the baseline cells and remains insignificant with controls and family fixed effects. Appendix~\ref{app:employmentrobust} reports the 40th- and 60th-percentile splits, terciles, quartiles, the two-dimensional surface, and a family-wise wild-cluster test around a flexible additive fit. The cluster-adjusted additive-null test gives $p=0.202$. Studentizing at the occupation level rather than at the level of the resampled clusters gives $p=0.065$. The resampling is at the family level, so the cluster-adjusted value is the appropriate one; the smaller figure is reported so that the choice is visible. The robust statement is that the decline is concentrated in the high-usage, high-execution corner and is compatible with additive negative associations.

\begin{table}[H]
\centering
\caption{Observed AI usage and the post-2022 employment change, cognitive occupations}
\label{tab:usage}
\small
\setlength{\tabcolsep}{4pt}
\begin{tabular}{lccccc}
\toprule
 & (1) & (2) & (3) & (4) & (5) \\
Dep.\ var.: $\Delta\ln$ emp.\ 2022--2025 & usage vs.\ cap. & interaction & cells & cells+controls & cells+group FE \\
\midrule
Observed usage (SD)                  & $-0.026^{***}$ & $-0.024^{**}$  &  &  &          \\
                                     & $(0.008)$      & $(0.009)$      &  &  &          \\
AI-capable share (SD in col.\ 1)     & $-0.008$       & $-0.121^{+}$   &  & $-0.124^{+}$ &          \\
                                     & $(0.013)$      & $(0.071)$      &  & $(0.069)$ &          \\
Execution share (SD)                 &                & $-0.028^{***}$ &  &  &          \\
                                     &                & $(0.010)$      &  &  &          \\
Execution $\times$ usage             &                & $-0.011$       &  &  &          \\
                                     &                & $(0.011)$      &  &  &          \\
Routine intensity                    &                & $-0.025^{***}$ &  & $-0.030^{***}$ &          \\
                                     &                & $(0.007)$      &  & $(0.007)$ &          \\
High execution, low usage            &                &                & $-0.035^{+}$ & $-0.024$ & $-0.005$ \\
                                     &                &                & $(0.020)$ & $(0.022)$ & $(0.023)$ \\
Low execution, high usage            &                &                & $-0.032$ & $-0.017$ & $-0.024$ \\
                                     &                &                & $(0.021)$ & $(0.021)$ & $(0.020)$ \\
High execution, high usage           &                &                & $-0.115^{***}$ & $-0.084^{***}$ & $-0.056^{***}$ \\
                                     &                &                & $(0.022)$ & $(0.020)$ & $(0.021)$ \\
\midrule
SOC major-group FE                   & no & no & no & no & yes \\
Observations                         & 364            & 364            & 364 & 364 & 364 \\
\bottomrule
\end{tabular}
\par\vspace{4pt}
\footnotesize\raggedright Note: OLS on $364$ cognitive occupations with 2022 and 2025 employment and an Anthropic Economic Index usage measure; robust standard errors are in parentheses. Usage is standardized within the cognitive sample, as is execution share in column 2. Column 1 jointly includes usage and the standardized AI-capable share; the Eloundou exposure alternative is reported in the text.
\par\smallskip Usage comes from \path{labor_market_impacts/job_exposure.csv} in the Hugging Face \path{Anthropic/EconomicIndex} dataset (released March 2026; accessed July 2026). The underlying conversations begin in December 2024, inside the outcome window, so usage is contemporaneous with the outcome.
\par\smallskip Column 2 adds execution share, the execution-by-usage interaction, capability, and routine intensity. Columns 3--5 use median-split cells, with the low-low cell omitted (raw means: $+6.9$, $+3.7$, $+3.4$, and $-4.6$ percent; $n=84$, $98$, $97$, and $85$). Column 4 adds capability and routine intensity; column 5 adds SOC major-group fixed effects. The super-additivity contrast is $-0.047$ ($p=0.14$), $-0.043$ ($p=0.16$), and $-0.026$ ($p=0.38$) in columns 3--5. $^{+}\,p<0.10$, $^{**}\,p<0.05$, $^{***}\,p<0.01$.
\end{table}

\subsection{Wages and seniority}
\label{sec:wagemargin}

The wage margin complicates a simple labor-demand interpretation. Replacing employment change in the controlled specification with the change in mean occupational wages produces a positive execution coefficient of $0.102$ (Table~\ref{tab:wage}). Execution-heavy occupations have lower employment growth but higher relative mean-wage growth.

\begin{table}[H]
\centering
\caption{The wage margin: log mean-wage change, cognitive occupations}
\label{tab:wage}
\small
\begin{tabular}{lc}
\toprule
Dep.\ var.: $\Delta\ln$ mean wage 2019--2025 & (1) \\
\midrule
AI-capable share $a_j$ & $-0.029$ \\
                       & $(0.029)$ \\
Execution share $e_j$  & $\phantom{-}0.102^{**}$ \\
                       & $(0.045)$ \\
Routine intensity      & $\phantom{-}0.010^{**}$ \\
                       & $(0.005)$ \\
Rate exposure          & $-0.003$ \\
                       & $(0.014)$ \\
\midrule
Observations           & 281 \\
\bottomrule
\end{tabular}
\par\vspace{4pt}
\footnotesize\raggedright Note: the specification of Table~\ref{tab:reg} column 3 with the log mean-wage change as outcome, OLS with robust standard errors in parentheses. $^{+}\,p<0.10$, $^{**}\,p<0.05$, $^{***}\,p<0.01$.
\end{table}

Several mechanisms can generate this combination: lower-paid execution workers may exit first, remaining work may be upgraded, scarce human execution may command a premium, or within-occupation task composition may change. OEWS mean wages cannot distinguish a price change for a fixed worker from changes in who remains employed. So the employment coefficients describe net occupational change rather than an identified inward demand shift.

The person-level CPS seniority test is null. The entry-age-by-post-2022 interaction is small, insignificant, and slightly positive for execution share, capability, and both jointly. The result is unchanged by occupation, year, occupation-by-entry, and entry-by-year fixed effects. Its standard error implies a minimum detectable effect of roughly $0.8$ percentage points on an outcome with mean $0.976$. The sample conditions on labor-force participation and observes neither hiring nor failure to enter an occupation. So the null rules out large effects on this narrow margin but does not test the entry-flow prediction directly. Appendix~\ref{app:cps} reports the full specifications.

\subsection{A within-venue case: candidate volume and reviewing effort at ICLR}
\label{sec:iclr}

The occupation data observe neither candidate volume $G_j$ nor verification capacity $V_{jt}$. One research venue observes both. The International Conference on Learning Representations (ICLR) posts every reviewed submission and every review on OpenReview. Machine learning is also a field in which AI lowered the cost of producing candidate papers early. I use the PaperCopilot mirror of these records \citep{papercopilot2026}, assembled for a companion study of academic evaluation \citep{gan2026certified}.

Figure~\ref{fig:iclr_effort} shows both sides of equation~\eqref{eq:verificationpressure}. Reviewed submissions rose from \ePapersNineteen{} in 2019 to \ePapersTwentySix{} in 2026, and reviews written rose from \eReviewsNineteen{} to \eReviewsTwentySix{}, a \eReviewsRatio-fold increase. Reviews per paper moved from \eRevPerNineteen{} to \eRevPerTwentyOne{} in 2021 and have stayed near four since. Total review text grew from \eTotalWordsNineteen{} million to \eTotalWordsTwentySix{} million words. Mean review length rose from \eWordsNineteen{} words in 2019 to a peak of \eWordsTwentyTwo{} in 2022, fell to \eWordsTwentyFour{} in 2024, and was \eWordsTwentySix{} in 2026. Mean reviewer confidence fell from \eConfNineteen{} to \eConfTwentySix{}, most of it by the 2023 cohort (\eConfTwentyThree{}). Part of the expansion was institutional: from 2025, ICLR required every submission to have an author registered to review, a rule absent from its 2024 call \citep{iclr2024cfp,iclr2025cfp}.

In the notation of equation~\eqref{eq:verificationpressure}, the venue raised $V_{jt}$, from 2025 partly by drafting authors as reviewers, and held reports per candidate near four. Capacity counted in reports kept pace with candidate volume. Per-report measures did not grow with volume after 2022: review length fell from its 2022 peak, and reviewer confidence fell. The exhibit observes $G$ and $V$ for one venue. It is not a causal test, because submission growth, the reviewing rule, and a randomized study of language-model review feedback in the 2025 cohort \citep{thakkar2026} coincide.

\begin{figure}[t]
\centering
\includegraphics[width=0.95\textwidth]{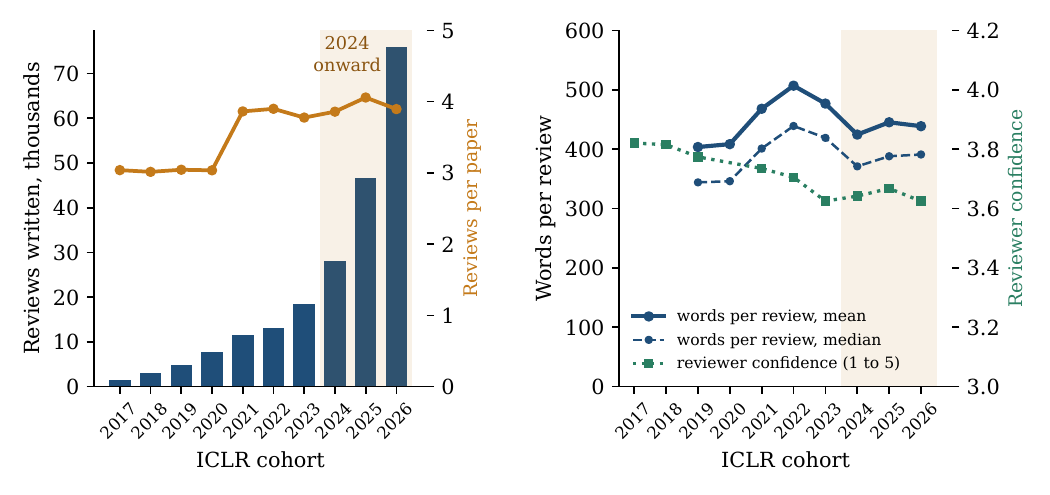}
\caption{Candidate volume and reviewing effort at ICLR, 2017--2026. Left: reviews written (bars) and reviews per paper (line). Right: mean and median words per review and mean reviewer confidence (1--5). The shaded area marks post-ChatGPT cohorts. Source: PaperCopilot mirror of OpenReview \citep{papercopilot2026}.}
\label{fig:iclr_effort}
\end{figure}

The same records allow a within-venue look at the mechanical share $m_j$, which scales the impact response in Proposition~\ref{prop:congestion}. Classified from their abstracts, about \hShareTheoryAll{} percent of submissions are proof-based, with claims that a reader must judge. About \hShareBenchAll{} percent are benchmark-based, with claims that could in principle be checked by re-running code. The classification agrees with authors' primary areas. Of learning-theory submissions, \hValTheoryLT{} percent are classified as proof-based and \hValBenchLT{} percent as benchmark-based; of datasets-and-benchmarks submissions, \hValBenchDB{} percent are benchmark-based and \hValTheoryDB{} percent proof-based.

Figure~\ref{fig:iclr_hetero} shows that the two groups moved together. For proof-based papers, the within-cohort rank correlation between a paper's mean review score and its citation percentile fell from \hRhoThPre{} in the 2019--2023 cohorts to \hRhoThPost{} in 2024--2025. For benchmark-based papers it fell from \hRhoBePre{} to \hRhoBePost{}. The difference-in-differences is \hRhoDid{} (standard error \hRhoDidSE), and the corresponding triple difference in the score--citation slope is \hSlopeTriple{} percentile points (\hSlopeTripleSE). Review length, within-paper score dispersion, mean scores, and acceptance rates show no differential change. Reviewer confidence is lower on proof-based papers, and the gap widened after 2023 (\hConfDid, standard error \hConfDidSE). But the gap was already widening from 2018 to 2023; measured against 2023 alone, the change is \hConfLastDid{} (\hConfLastDidSE). The youngest cohorts have had little time to accumulate citations, so late correlations are comparable across groups within a cohort, not in level across cohorts.

\begin{figure}[t]
\centering
\includegraphics[width=0.95\textwidth]{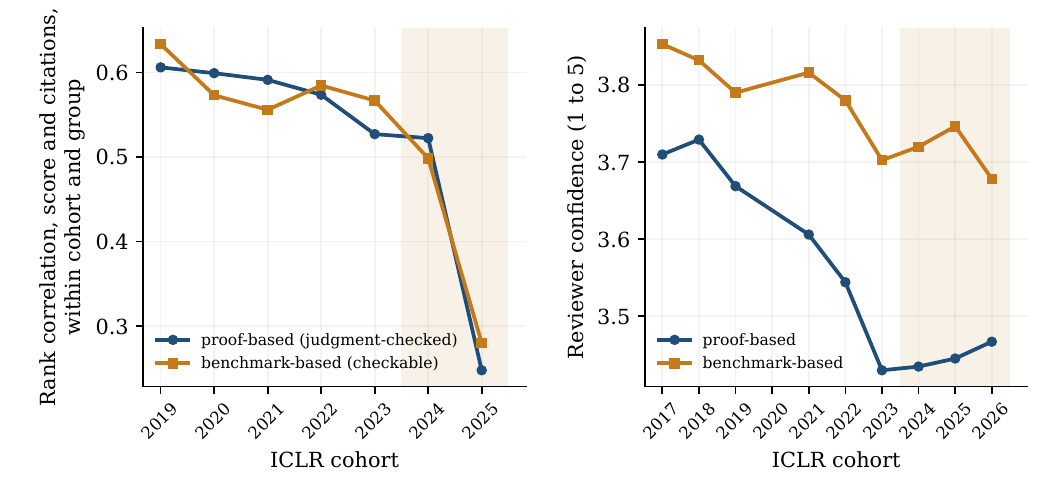}
\caption{Proof-based and benchmark-based submissions at ICLR. Left: within-cohort rank correlation between mean review score and citation percentile. Right: mean reviewer confidence. Groups are classified from abstracts. The shaded area marks post-ChatGPT cohorts. Source: PaperCopilot mirror of OpenReview \citep{papercopilot2026}.}
\label{fig:iclr_hetero}
\end{figure}

The null has a direct reading in the model. The venue reads both kinds of paper; it does not re-run benchmark submissions. So mechanical checking is close to zero for both groups, $m_j\approx 0$, and equation~\eqref{eq:verificationpressure} predicts a common strain, which is what the records show. The case bounds the venue's use of mechanical checking. It does not test the contrast across domains in Proposition~\ref{prop:congestion}, which requires a setting where checking is in fact mechanized.

\subsection{Findings and limits}

Table~\ref{tab:status} collects the empirical objects above with the interpretation each one supports, separating what the data establish from what remains a stated limit.

\begin{table}[H]
\centering
\caption{Empirical findings and their interpretation}
\label{tab:status}
\small
\begin{tabular}{@{}>{\raggedright\arraybackslash}p{0.43\linewidth}>{\raggedright\arraybackslash}p{0.49\linewidth}@{}}
\toprule
Empirical object & Result and interpretation \\
\midrule
Execution share as a task-level measure & Human-benchmarked on $200$ blinded tasks; reproducible across model coders; distinct from capability and routine intensity \\
Increment beyond O*NET evaluation activities & Modest and control-dependent; the main value is common task-level architecture and portability \\
Execution-employment gradient & Present before generative AI; no detected break; largely between occupational families \\
Within-occupation adoption gradient & Weak: significant for one predetermined sector measure in 2022--2025 only, with a similar insignificant pre-2019 estimate and none over 2019--2025; absent within sectors in 2022--2025; not separable from software adoption \\
Eloundou exposure gradient & Negative before 2022 and significantly steeper in 2022--2025 than in the pandemic window; dated, not causal \\
Usage--execution concentration & High-high cell contracted; continuous and family-adjusted super-additivity tests are insignificant \\
Wage margin & Higher relative mean-wage growth despite lower employment growth, inconsistent with a simple one-margin demand interpretation \\
Seniority margin & CPS null on employment conditional on labor-force participation; hiring and occupational entry remain unobserved \\
Candidate volume and reviewing effort (ICLR) & Reviews kept pace with submissions at about four per paper; review length fell from its 2022 peak and reviewer confidence fell; a single-venue case, not a causal test \\
Proof-based versus benchmark-based papers (ICLR) & No differential change in score--citation association between the groups; consistent with $m_j\approx 0$ at the venue; does not test the contrast across domains \\
\bottomrule
\end{tabular}
\end{table}

\section{Discussion and Conclusion}
\label{sec:conclusion}
\label{sec:agenda}

This paper separates two questions that one-dimensional exposure measures combine: can AI perform an occupation's tasks, and does the remaining human role consist primarily of producing output or evaluating it? Scoring both dimensions on the same $19{,}265$ O*NET task statements creates a common task-level architecture. The model then uses that architecture to order the evidence by horizon.

At the first stage, a capability improvement should reduce conditional human labor demand more in execution-heavy roles, particularly at entry. The occupation data contain the required cross-sectional variation but not a credible AI treatment. The execution gradient begins well before generative AI, has no detected post-2022 break, and lies largely between occupational families. Within occupations, prior AI adoption predicts a steeper gradient for one predetermined sector measure in 2022--2025 only, with a similar estimate before 2019 and no separation from software adoption. It does not appear across industries within a sector in 2022--2025. The Eloundou exposure gradient is steeper in 2022--2025 than in the pandemic window, but the design can date that change without attributing it. The high-use, high-execution cell contracts, while the continuous interaction is imprecise and the CPS seniority test is null. These are descriptive checks on the model's ordering, not estimates of its causal parameters.

At the second stage, candidate volume can rise against predetermined verification capacity. Equation~\eqref{eq:verificationpressure} predicts immediate pressure where checking depends on scarce human capacity and little pressure where checking is mechanically scalable. The occupation data do not observe candidate volume or verification capacity. The ICLR case in Section~\ref{sec:iclr} observes both for one venue. Reviews kept pace with submissions at about four per paper, while review length fell from its 2022 peak and reviewer confidence fell. Proof-based and benchmark-based submissions show no differential change in score--citation association, as expected where the venue reads both kinds of paper rather than re-running one. The case documents the second stage at one venue; it does not identify a causal effect or test the contrast across domains that differ in mechanical checking.

At the third stage, the evaluator pipeline responds only after the affected entry cohort matures. A decline in senior numbers is common to models in which entry feeds the future workforce. The sharper implication is a decline in evaluative skill per senior and in downstream evaluation performance conditional on evaluator numbers, because the removed input is execution practice. Current data end before either lagged outcome can be observed.

The next empirical design should follow the model's sequence rather than add another occupation cross-section. It needs four things. The first is adoption measured before the shock at the firm rather than the industry level, or a plausibly exogenous capability shock. The second is hiring, vacancy, and task-assignment flows that observe the entry response. The third is aggregate output and verification loads rather than own-output text alone. The fourth is linked cohorts that connect execution practice to later evaluator stocks and performance. The task-level human audit now benchmarks the execution measure against independent judgment, but it supplies none of these treatment, flow, verification-load, or cohort outcomes.

The paper's contribution is neither a causal estimate of AI displacement nor the generic claim that verification may become scarce. It is a measurable execution-evaluation axis and an explicit sequence of testable implications that distinguishes immediate verification pressure from delayed human-capital depletion. The evidence in hand partially examines the early links. The long-run evaluator-pipeline prediction remains a research hypothesis.

\bibliographystyle{aer}
\bibliography{reference}

\clearpage
\appendix
\section{Detailed Derivations, Quantitative Extension, and Comparative Predictions}
\label{app:formal}
\label{sec:theory}

\subsection{CES unit cost and capability response}

Section~\ref{sec:framework} states the execution technology in equation~\eqref{eq:Y}. This appendix supplies the cost-minimization details. Frontier capability $A_t$ is an efficiency shifter on the AI input, so a rise in $A_t$ is equivalent to a fall in the effective AI price $q/(\theta_j A_t)$. This keeps frontier capability, role exposure, and the purchased quantity of AI input distinct.

Cost minimization for one unit of execution output gives, by Shephard's lemma,
\begin{equation}
s_j^*=\alpha^\sigma\left(\frac{c_j}{w}\right)^\sigma,
\qquad
c_j=\left[\alpha^\sigma w^{1-\sigma}+(1-\alpha)^\sigma\left(\frac{q}{\theta_j A_t}\right)^{1-\sigma}\right]^{1/(1-\sigma)},
\label{eq:sstar}
\end{equation}
where $c_j$ is the CES unit cost of execution. Define
\[
\psi_j^E=\frac{\kappa p_j^{\sigma-1}}{1+\kappa p_j^{\sigma-1}},
\qquad
\kappa=\left(\frac{\alpha}{1-\alpha}\right)^\sigma,
\qquad
p_j=\frac{q}{w\theta_j A_t}.
\]
Then $\psi_j^E$ is the human cost share of execution, and differentiating equation~\eqref{eq:sstar} gives equation~\eqref{eq:psi}. The conditional human labor quantity falls as capability rises for any $\sigma>0$. The restriction $\sigma>1$ adds that the human cost share falls and makes it more plausible that substitution dominates output expansion.

The labor quantity $s_j^*$ should not be confused with the cost share $\psi_j^E$. Equation~\eqref{eq:H} aggregates labor quantities across execution and evaluation. The cost share enters only through the capability elasticity.

\subsection{Proof and interpretation of Proposition~\ref{prop:sorting}}

\begin{proof}[Proof of Proposition~\ref{prop:sorting}]
From equation~\eqref{eq:psi}, $\partial s_j^*/\partial A_t<0$, with proportional size governed by $\theta_j$ and so by $a_j$. From equation~\eqref{eq:H},
\[
\frac{\partial H_j}{\partial A_t}=e_j\frac{\partial s_j^*}{\partial A_t},
\]
which is strictly more negative as $e_j$ rises at fixed $a_j$. For the proportional decline,
\[
\frac{\partial\ln H_j}{\partial A_t}=\underbrace{\frac{e_j s_j^*}{e_j s_j^*+(1-e_j)\underline r}}_{f_j}\frac{\partial\ln s_j^*}{\partial A_t}.
\]
Because $\underline r>0$, $f_j$ is strictly increasing in $e_j$. The second term is negative, so the magnitude of the proportional decline also rises with $e_j$.
\end{proof}

Proposition~\ref{prop:sorting} compares the model with a capability-only exposure measure. Such a measure orders tasks by automatability and contains no independent $e_j$. So it cannot generate an execution-share coefficient after capability is held fixed unless an additional task attribute is introduced. This does not distinguish the framework from the full class of task models, which can incorporate task-specific elasticities, regulation, complementarities, or evaluation tasks directly.

The product form in equation~\eqref{eq:Dstar} explains why the interaction is more structural than the additive capability coefficient. Conditional on output, higher capability lowers human execution demand. In the observed occupation stock, output expansion and evaluation complementarity can offset that effect. Capability also operates only through the execution component and vanishes as $e_j\rightarrow0$. The framework signs $\beta_e$ and $\beta_{ea}$ in equation~\eqref{eq:reg}, while leaving the reduced-form sign of $\beta_a$ to the data.

The interaction is difficult to estimate in a single cross-section. Within cognitive occupations, $e_j$ and $a_j$ are moderately correlated, and their product is highly collinear with the levels. The interaction and super-additivity tests enter with the predicted sign but are imprecise. Hiring or vacancy flows with predetermined adoption variation would provide a cleaner test of equation~\eqref{eq:Dstar}.

Two maintained restrictions in equation~\eqref{eq:H} deserve emphasis. The oversight residual $\underline r$ is common across occupations, although licensing, liability, regulation, certification, and organizational design are likely to make it vary. Role-level output is also held fixed, although cheaper execution may expand output and raise demand for both functions. So the employment coefficients are net occupational changes rather than pure conditional displacement.

\subsection{Proof and extended implications of Proposition~\ref{prop:pipeline}}

\begin{proof}[Proof of Proposition~\ref{prop:pipeline}]
The immediate entry effect follows from $N_t=\nu(c_t^N)$ with $\nu$ increasing. For the first $L$ periods, equation~\eqref{eq:senior} continues to depend on the pre-shock novice mass, so the path of $S_t$ is the same as under the no-shock counterfactual. Once the lower inflow reaches the recursion, $S_{t+1}=(1-\delta_S)S_t+\mu N_\infty$. This linear difference equation converges to $S_\infty=\mu N_\infty/\delta_S<\mu N_0/\delta_S$ at rate $1-\delta_S$.
\end{proof}

Assumption~\ref{ass:seniority} cannot be inferred from the occupation-level execution share because the score does not vary by seniority within an occupation. A direct test would require occupation-by-seniority task descriptions from job postings, O*NET experience and on-the-job-training zones, or administrative seniority labels.

\begin{corollary}[Career-stage gradient]
\label{cor:gradient}
Under Assumption~\ref{ass:seniority}, entry-level displacement exceeds senior displacement at a given AI-capable share, and the execution gradient is steeper at entry.
\end{corollary}

Aggregate evaluation capacity is
\begin{equation}
V_t^{\mathrm{tot}}=S_t\bar s_t^V,
\label{eq:Vtot}
\end{equation}
where $S_t$ is the senior mass and $\bar s_t^V$ is average evaluation skill among seniors. This closure is mass times average skill rather than a complete overlapping-cohorts model.

\begin{corollary}[Downstream evaluation degradation]
\label{cor:downstream}
After the cohort lag, the lower novice inflow reduces $S_t$, and cohorts trained with less execution practice reduce $\bar s_t^V$. Total evaluation capacity therefore falls through both the quantity and quality margins. Fields that lose the entry rung earlier should exhibit weaker mid-career evaluation performance earlier.
\end{corollary}

The corollary is conditional on two responses the model does not derive: AI must reduce entry-level human execution, and execution practice must be an input into evaluative skill. The current data do not observe the second response and address the first only weakly. So Proposition~\ref{prop:pipeline} is cohort accounting conditional on a pipeline shock, not a result that AI causes one.

Any framework in which entry feeds the future workforce can produce a delayed decline in the number of senior workers. The more specific prediction here concerns quality. Because the removed input is execution practice, average evaluative skill and downstream evaluation outcomes may decline even at a given number of senior evaluators. Section~\ref{app:predictions} develops the comparison formally.

\subsection{A quantitative structure: uses and limits}
\label{sec:structural-form}

The analytical aggregator in \eqref{eq:H} protects evaluation labor through a fixed oversight residual $\underline r$. For quantitative illustration, I replace that residual with a constant-elasticity complementarity between execution services $X_j$ and evaluation services $V_j$:
\begin{equation}
Y_j=\Big[\omega_j^{1/\eta}X_j^{(\eta-1)/\eta}+(1-\omega_j)^{1/\eta}V_j^{(\eta-1)/\eta}\Big]^{\eta/(\eta-1)},\qquad V_j=s^V_j.
\label{eq:upper}
\end{equation}
The execution service $X_j$ is produced by the lower nest in \eqref{eq:Y}, which combines human execution and AI with elasticity $\sigma>1$. At the upper nest, $\eta<1$ captures complementarity: reliable output cannot be increased simply by generating more and checking less.

Three calibration conventions should be explicit. First, the technology weight equals the measured execution share, $\omega_j=e_j$. Because $e_j$ indexes task statements rather than an expenditure, time, or employment share, this is a convention, not an aggregation result. Second, setting $V_j=s^V_j$ makes the upper nest a per-worker illustration, with no evaluation wage, hiring margin, or firm choice over evaluation intensity. The aggregate service consistent with \eqref{eq:Vtot} would be $V_{jt}=S_{jt}\bar s^V_{jt}$, and the scenario plots only the headcount component $S_t$. Third, the AI-capable share affects both the level of the effective AI price through $\theta_j$ and the rate at which the post-2022 shock lowers that price. A richer design would identify these two roles with separate variation. Two further limits apply. Equation~\eqref{eq:Vtot} combines senior mass and average skill but is not a full cohort model. And the entry response $N_t=\nu(c^N_t)$ is assumed rather than derived from novice labor demand. So Proposition~\ref{prop:pipeline} is cohort accounting conditional on a fall in entry-level practice, not evidence that AI causes that fall.

Evaluation enters on two margins, and both are calibrated rather than estimated. The first is production insulation: because $\eta<1$ and AI does not enter $V_j$, evaluation-heavy roles should respond less to cheaper execution. The second is pipeline quality. Equations~\eqref{eq:sV} and \eqref{eq:senior} imply that reduced junior execution lowers future evaluation capacity after the cohort lag $L$, the source of predictions P3 and P4. The cross-occupation employment response that could discipline the first margin is flat in these data, and the second margin is unobserved. So $\eta$ is calibrated rather than estimated, and the apparatus yields no structural estimate (see Appendix~\ref{sec:estimation}). The two-nest form also clarifies pooled factor-demand estimates. A one-nest regression that does not separate the two functions may return an elasticity below one, suggesting complementarity between human labor and AI. The estimate in fact blends the substitutable execution margin ($\sigma>1$) with the complementary evaluation margin ($\eta<1$), under occupation-specific weights related to $e_j$.

\subsection{Predictions and comparison with alternative mechanisms}
\label{app:predictions}

\begin{description}[leftmargin=2.2em,style=nextline,font=\normalfont\itshape]
\item[P0. Impact verification pressure.] If cheaper execution expands candidate volume while verification capacity is predetermined, the checked share falls on impact in the binding region. The effect is concentrated where mechanical verification covers a smaller share of output.
\item[P1. Execution-share sorting.] Conditional on the AI-capable share, displacement increases in execution share: $\beta_e<0$ in \eqref{eq:reg} with the signed-change outcome. A capability-only exposure measure predicts no such independent effect once capability is held fixed; the execution-evaluation mechanism predicts one because evaluation carries a human oversight residual.
\item[P2. Career-stage gradient.] Within AI-exposed occupations, entry-level displacement exceeds senior displacement, and the execution gradient is steeper at entry.
\item[P3. Cohort-lag temporal sequence.] Entry employment falls immediately, the senior stock is preserved for $L$ periods, and the senior stock declines geometrically thereafter; this is cohort accounting conditional on the entry and practice responses to AI, which the model assumes rather than derives (Section~\ref{sec:structural-form}).
\item[P4. Downstream evaluation degradation.] Fields that lost the entry-level rung earliest exhibit weaker mid-career evaluation capacity, with cohort-lag onset.
\end{description}

No individual prediction in this list is claimed to be unique. P0 overlaps with recent verification-bottleneck models \citep{catalini2026agi}; P3 overlaps with apprenticeship and learning-by-doing models \citep{garicano2025training,afrouzi2026automation,dogonowski2026apprenticeship,li2026generation}. The more specific contribution is the joint sequence and its measurable first-stage task attribute. So Table~\ref{tab:compare} compares the labor-market closures behind P1--P4 rather than entire literatures. The task-based column represents a capability-only exposure implementation. A general task model can incorporate evaluation tasks, task-specific elasticities, regulation, and complementarities, and can generate an execution-share effect. P1 is distinctive only relative to the one-dimensional automatability measures commonly used in empirical work.

With that qualification, the near-term predictions separate mechanisms only partially. P2 distinguishes the career-stage mechanisms from a capability-only exposure index but not from one another. P1 separates the present model from that index and from the other mechanisms, to the extent that execution share is not collinear with knowledge rank or codifiability. Section~\ref{sec:results} shows that the required variation exists.

The long-run predictions differ by outcome. A delayed decline in the aggregate number of senior evaluators (P3a) follows from any model in which entry feeds the future stock, including dynamic tacit-knowledge and apprenticeship-viability accounts. A decline in evaluation quality per senior (P3b), and weaker evaluation outcomes at a given number of evaluators (P4), require the practice-to-skill channel made explicit here. The contribution is not the generic possibility of a pipeline. It is the measurable execution-evaluation axis and the separation of immediate pressure from later evaluator quantity and quality. P1 and P2 can be examined now, P0 can be tested with aggregate generation and verification loads, and P3 and P4 occur only after the cohort lag.

The placement of the rival models in the P3 rows depends on how senior capability is produced. In knowledge-hierarchy models, it is not produced by entry-rung work. \citet{ide2025ai} treat knowledge as a fixed endowment that the hierarchy reallocates, so automating worker-level tasks does not reduce the future supply of high-knowledge agents. In \citet{garicano2000hierarchies}, knowledge is acquired through costly investment separate from work at the lower rung; automating routine problems may even raise the return to acquiring exception-handling knowledge. Neither closure generates a delayed decline in senior capacity.

Tacit knowledge is less clear-cut. In a static formalization it is a fixed skill attribute, so senior workers remain protected. In its usual substantive interpretation, however, tacit knowledge is acquired through experience. A dynamic version would route senior capability through the same practice that AI removes and would converge toward the mechanism developed here. The distinction then lies in the explicit execution-evaluation measure and in the modeled link from practice to later evaluation skill. Apprenticeship viability also lets entry feed the future stock through on-the-job training, which is why it predicts P3a. Its primitives do not by themselves imply the per-worker quality effect in P3b.

The apprenticeship-viability account is especially close, but the distinction is one of margins rather than mutually exclusive theories. A firm-investment mechanism predicts less training when AI absorbs entry-level tasks. The practice-to-skill mechanism predicts weaker evaluative judgment when AI performs the underlying practice, even if formal training continues. So P4 is the sharper discriminating outcome because it concerns the quality of the human capital produced, not only the number of juniors trained.

\citet{ide2025ai} generate the same near-term pattern through general equilibrium: autonomous AI replaces worker-level labor and reallocates the remaining workforce toward problem solving. The present model differs only in its delayed predictions for the senior stock and evaluation quality. Firm organization, wages, and compute prices lie outside this paper's cross-sectional labor-demand comparison.

\begin{table}[t]
\centering
\caption{Which formal models predict each empirical pattern, under which assumption}
\label{tab:compare}
\footnotesize
\setlength{\tabcolsep}{3pt}
\begin{tabular}{>{\raggedright\arraybackslash}p{4.5cm}ccccc}
\toprule
 & Capability- & Knowledge & Dynamic tacit & Apprentice- & Execution-- \\
 & only & hierarchy & knowledge & ship & evaluation \\
\addlinespace[2pt]
\emph{Key assumption:} & \emph{\scriptsize one automat-} & \emph{\scriptsize knowledge not} & \emph{\scriptsize tacit skill grows} & \emph{\scriptsize training is a} & \emph{\scriptsize eval.\ skill from} \\[-2pt]
 & \emph{\scriptsize ability index} & \emph{\scriptsize made at entry} & \emph{\scriptsize with experience} & \emph{\scriptsize firm choice} & \emph{\scriptsize exec.\ practice} \\
\midrule
Entry-level decline (gradient)      & No & Yes & Yes & Yes & Yes \\
P1: $\beta_e<0 \mid a_j$             & No & Partial & Partial & Partial & Yes \\
P2: steeper gradient at entry        & No & Yes & Yes & Yes & Yes \\
P3a: aggregate senior \emph{quantity} falls after $L$ & No & No & Yes & Yes & Yes \\
P3b: per-worker senior eval.\ \emph{quality} falls & No & No & Partial & Partial & Yes \\
P4: eval.\ outcomes decay at given numbers & No & No & Partial & No & Yes \\
\bottomrule
\end{tabular}
\par\vspace{4pt}
\footnotesize\raggedright Note: Columns represent formal model closures under the stated assumptions, not entire literatures. ``Capability-only exposure'' refers to \citet{acemoglu2022tasks} implemented with a one-dimensional automatability index; the ``No'' entries mean that this implementation has no independent execution-share term.
\par\smallskip ``Knowledge hierarchy'' combines \citet{garicano2000hierarchies} and \citet{ide2025ai}. Neither routes senior capability through entry-rung work, so both P3 margins are No. ``Dynamic tacit knowledge'' treats skill as experience-produced: it predicts fewer experienced workers after entry contracts (P3a) and reaches P3b/P4 only when the accumulated skill is evaluative. ``Apprenticeship viability'' \citep{garicano2025training} predicts fewer future seniors (P3a) but does not derive evaluation quality at fixed numbers.
\par\smallskip P3a is the quantity margin; P3b and P4 are quality margins. ``Partial'' denotes a prediction that requires an auxiliary assumption.
\end{table}

\section{Additional Measurement Details and Descriptive Evidence}
\label{app:measurement}

\subsection{Weighting, batch order, and score uncertainty}

The baseline execution share gives core tasks twice the weight of supplemental tasks. Equal weights, O*NET task-importance weights, and expected-frequency weights yield occupation scores correlated $0.98$ to $0.99$ with the baseline among cognitive occupations. The six-year employment coefficient ranges from $-0.70$ to $-0.74$ across the four aggregations, with nearly unchanged $t$-statistics.

Production scoring used batches of $25$ statements in fixed O*NET order, so adjacent statements often came from the same occupation even though titles were hidden. I rescored the $500$ audited statements in three arms: production-style adjacency, randomized order with at most one statement per occupation per batch, and an identical repeat of the randomized arm. The repeated randomized inputs correlate $0.989$ and differ by $2.0$ points on average, showing that temperature zero is not perfectly deterministic. Grouped and randomized scores correlate $0.95$ and differ by $5.8$ points on average. Same-occupation adjacency pulls scores mildly toward the occupation mean, but the mean shift is only $-0.9$ points and the resulting systematic shift in the occupation-level execution share is $-0.011$. Drift relative to production scores generated months earlier is similar: mean absolute difference $6.5$ points and mean shift $+2.1$ points.

The employment regressions treat the generated scores as fixed. So their standard errors do not propagate scoring uncertainty. In the model-coder audit, the consensus differs from the production score by $10.1$ points on average. In the separate random human audit, the production score differs from the four-coder mean by $12.3$ points, while the mean absolute difference between individual human coders is $23.8$ points. These discrepancies are not classical measurement error, and the human mean is not ground truth, but they provide a scale for generated-regressor uncertainty. Repeated-score or errors-in-variables estimates are natural extensions.

\subsection{Sample flow and unconditional employment relationships}

Table~\ref{tab:flow} traces the sample from the scored O*NET universe to the estimation samples, and Figure~\ref{fig:dlog} plots the unconditional six-year employment relationships.

\begin{table}[t]
\centering
\caption{Sample flow}
\label{tab:flow}
\small
\begin{tabular}{lrr}
\toprule
 & All & Cognitive \\
\midrule
O*NET occupations scored (index data)                       & 798 & 400 \\
OEWS analysis universe                                       & 861 & 435 \\
\quad with index score                                       & 772 & 389 \\
\quad with 2019--2025 employment change and index score      & 708 & 337 \\
\quad\quad and routine intensity and rate exposure           &     & 284 \\
\quad with 2022--2025 employment change                      &     & 411 \\
\quad\quad and observed-usage exposure                       &     & 364 \\
\quad present in all three windows (balanced)                &     & 336 \\
\quad\quad and routine intensity                             &     & 334 \\
\bottomrule
\end{tabular}
\par\vspace{4pt}
\footnotesize\raggedright Note: cognitive means the ten white-collar SOC major groups ($11$, $13$, $15$, $17$, $19$, $23$, $25$, $27$, $29$, $43$). The OEWS universe exceeds the matched index count because some OEWS occupations have no scored O*NET counterpart; the 2022--2025 sample exceeds the 2019--2025 sample because 2018-SOC codes phased into separate OEWS estimation after May 2019 have no 2019 baseline.
\end{table}

\begin{figure}[t]
\centering
\includegraphics[width=\textwidth]{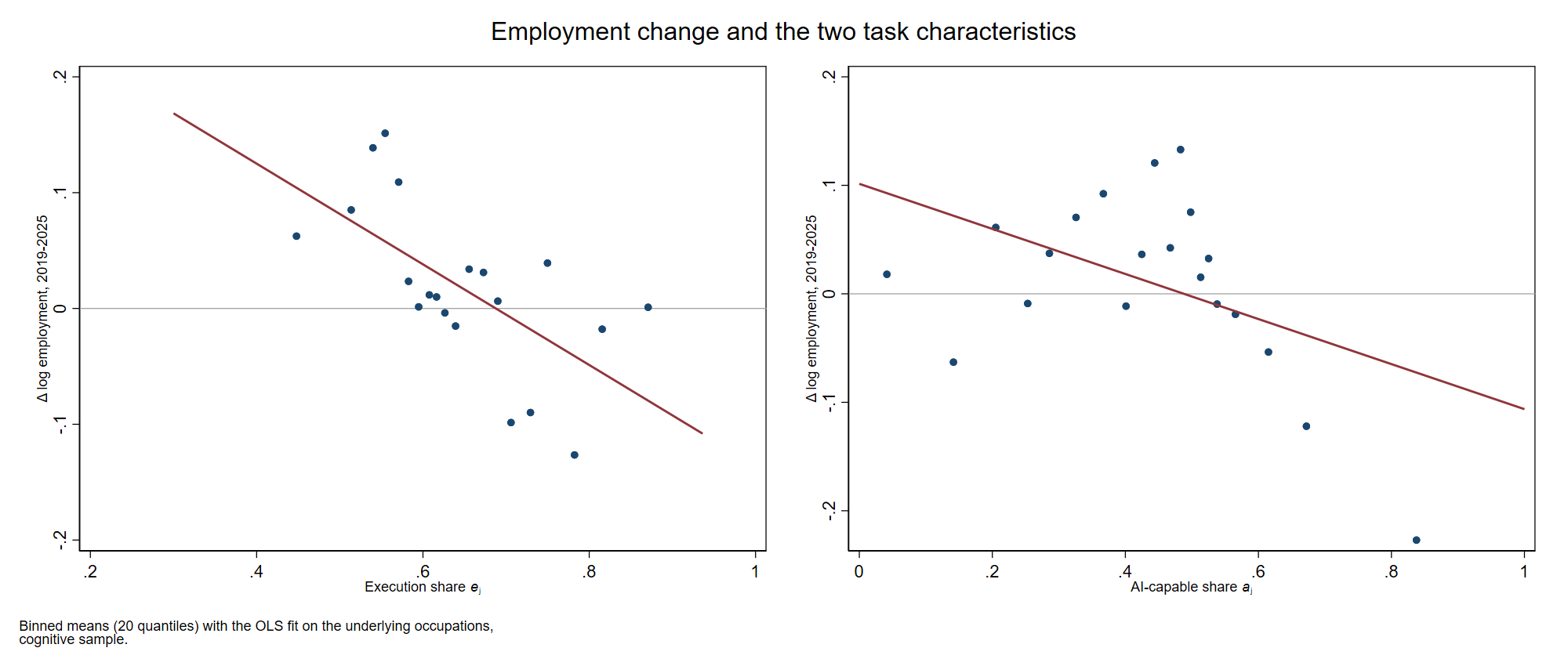}
\caption{Six-year log employment change against execution share (left) and AI-capable share (right), with binned means and linear fits for cognitive occupations. Section~\ref{sec:results} estimates the two relationships jointly.}
\label{fig:dlog}
\end{figure}

\subsection{Off-diagonal occupations and window-specific routine controls}

Table~\ref{tab:offdiag} lists the occupations on which the execution share and routine intensity diverge most sharply, in both directions.

\begin{table}[t]
\centering
\caption{Execution share and routine intensity diverge: occupations off the diagonal}
\label{tab:offdiag}
\small
\begin{tabular}{llcc}
\toprule
SOC & Occupation & $e_j$ & RTI \\
\midrule
\multicolumn{4}{l}{\emph{A. Non-routine execution (high $e_j$, low RTI)}}\\
29-1022 & Oral and maxillofacial surgeons & $0.81$ & $-0.96$ \\
27-1023 & Floral designers & $0.80$ & $-1.05$ \\
15-2021 & Mathematicians & $0.73$ & $-1.69$ \\
19-2012 & Physicists & $0.66$ & $-2.77$ \\
\midrule
\multicolumn{4}{l}{\emph{B. Routine evaluation (low $e_j$, high RTI)}}\\
43-9081 & Proofreaders and copy markers & $0.48$ & $\phantom{-}1.58$ \\
27-2023 & Umpires, referees, and sports officials & $0.49$ & $\phantom{-}0.48$ \\
13-2011 & Accountants and auditors & $0.43$ & $\phantom{-}0.19$ \\
13-2053 & Insurance underwriters & $0.30$ & $-0.13$ \\
\bottomrule
\end{tabular}
\par\vspace{4pt}
\footnotesize\raggedright Note: $e_j$ is the execution share ($0$ to $1$); RTI is standardized routine-task intensity, negative below average. Occupations are those with the largest standardized $e_j$-minus-RTI gap in each direction within cognitive occupations.
\end{table}

On the balanced cognitive sample with routine intensity ($N=334$), the univariate routine coefficient is $-0.016$ ($p=0.07$) in 2019--2022, $-0.031$ ($p<0.01$) in 2022--2024, and $-0.006$ in 2024--2025. When routine intensity, execution share, and capability enter jointly, routine intensity is insignificant in the first and last windows and equals $-0.022$ ($p<0.01$) in 2022--2024. The execution coefficients remain negative and significant in all three windows ($-0.27$, $-0.25$, and $-0.18$; all $p<0.05$). These regressions show that the execution gradient is not mechanically absorbed by routine intensity, although neither measure identifies a causal technology shock.

\section{Additional Employment Robustness}
\label{app:employmentrobust}

\subsection{Post-2022 thresholds, surface, and family-wise inference}

The median cells in Table~\ref{tab:usage} are one descriptive representation of the usage-execution plane. Moving both thresholds to the 40th percentile gives a high-high coefficient of $-0.082$ ($t=-3.5$), while the 60th-percentile split gives $-0.135$ ($t=-5.2$). The corner contrast becomes larger as the comparison sharpens: $-0.115$ at the median, $-0.143$ ($t=-3.7$) for top-versus-bottom terciles, and $-0.187$ ($t=-3.7$) for top-versus-bottom quartiles. These are variations on one descriptive object rather than independent tests. Figure~\ref{fig:surface} shows the full two-dimensional pattern behind them.

\begin{figure}[t]
\centering
\includegraphics[width=\textwidth]{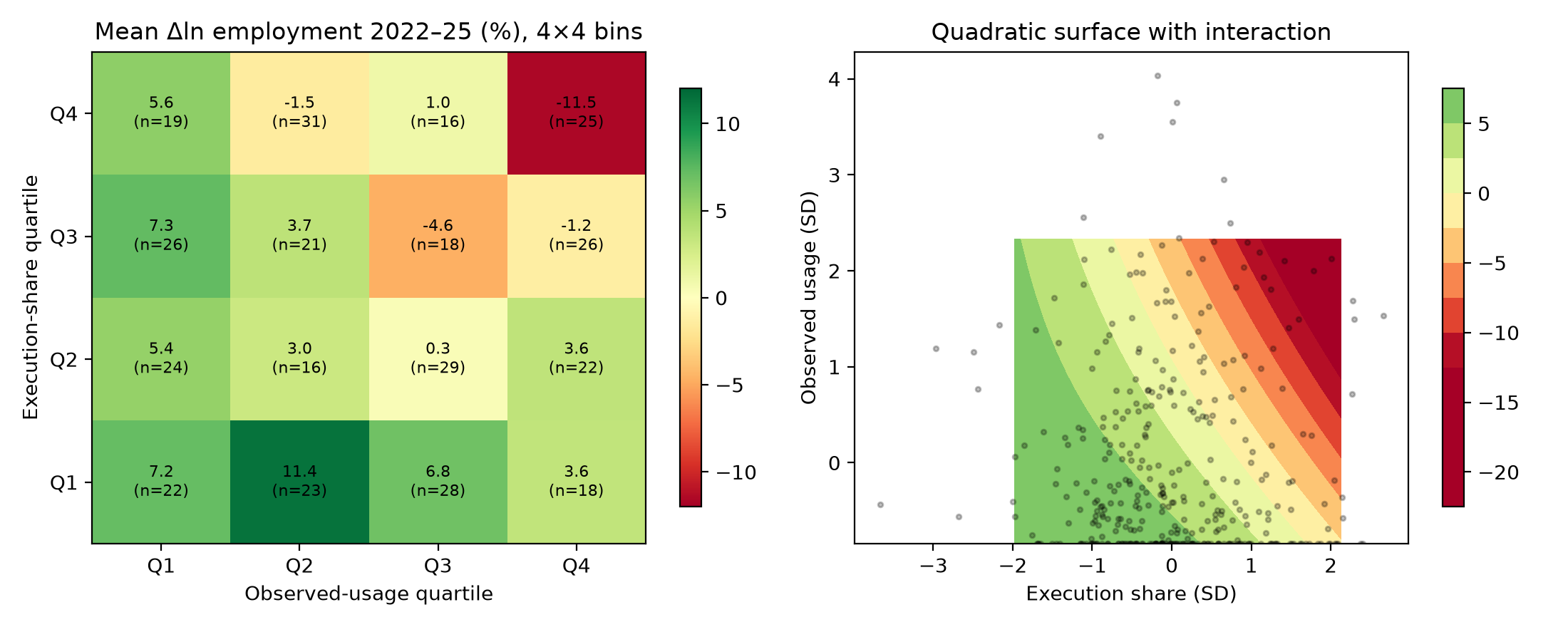}
\caption{Post-2022 employment change over the usage-execution plane. Left: mean 2022--2025 log employment change (percent) in rank-quartile cells, with cell sizes. Right: a quadratic surface with an interaction term. The joint-high cell is the deepest, but the continuous interaction ($t=-1.3$) and the family-wise additive-null test ($p=0.202$) are insignificant.}
\label{fig:surface}
\end{figure}

I use two permutation procedures over the family of five thresholds. A global no-association test permutes the raw outcome across occupations $2{,}000$ times and records the maximum corner $|t|$. The observed value $5.33$ exceeds every draw, yielding $p=0.0005$ by the plus-one convention. This test shows that the plane predicts employment change after threshold search, but it does not test interaction because additive negative main effects also produce a deep high-high corner.

The second test imposes the additive null. I fit cubic functions of execution and usage and assign wild Rademacher signs to residuals by family. I enumerate all $2^{10}=1{,}024$ assignments and studentize each super-additivity contrast with a family cluster-robust variance. The observed maximum $|t|$ is $2.43$; $207$ assignments are at least as extreme, for exact family-wise $p=0.202$. So the super-additive component is insignificant once dependence and studentization are imposed at the family level.

Office and administrative support accounts for $30$ of the $85$ high-high occupations. Splitting that cell gives a coefficient of $-0.192$ ($t=-5.0$) for SOC 43 and $-0.075$ ($t=-3.5$) for the remaining occupations. Dropping SOC 43 leaves a high-high coefficient of $-0.075$ ($t=-3.6$). Office support is an important contributor but not the entire result.

\subsection{Weighting and influence}

\begin{table}[t]
\centering
\caption{Weighted and influence-robust estimates of the headline regression}
\label{tab:weights}
\small
\begin{tabular}{lccc}
\toprule
 & (1) & (2) & (3) \\
Dep.\ var.: $\Delta\ln$ emp.\ 2019--2025 & unweighted & emp.-weighted & drop smallest decile \\
\midrule
AI-capable share $a_j$ & $-0.468^{***}$ & $-0.240^{+}$   & $-0.377^{***}$ \\
                       & $(0.124)$      & $(0.127)$      & $(0.122)$ \\
Execution share $e_j$  & $-0.701^{***}$ & $-0.909^{***}$ & $-0.728^{***}$ \\
                       & $(0.192)$      & $(0.216)$      & $(0.201)$ \\
Routine intensity      & $-0.051^{***}$ & $-0.040^{**}$  & $-0.053^{***}$ \\
                       & $(0.016)$      & $(0.019)$      & $(0.017)$ \\
Rate exposure          & $-0.025$       & $-0.094$       & $-0.063$ \\
                       & $(0.061)$      & $(0.066)$      & $(0.055)$ \\
\midrule
Observations           & 284            & 284            & 260 \\
$R^2$                  & 0.17           & 0.34           & 0.17 \\
\bottomrule
\end{tabular}
\par\vspace{4pt}
\footnotesize\raggedright Note: the specification of Table~\ref{tab:reg} column 3 on cognitive occupations, robust standard errors in parentheses. Column 1 repeats the unweighted baseline; column 2 weights by May 2019 employment; column 3 drops the bottom decile of 2019 employment. Leaving out one SOC major group at a time bounds the execution coefficient between $-0.790$ and $-0.487$, significant in every replicate. $^{+}\,p<0.10$, $^{**}\,p<0.05$, $^{***}\,p<0.01$.
\end{table}

Weighting by the inverse sampling variance implied by published OEWS relative standard errors gives an execution coefficient of $-1.03$ ($t=-4.8$). The AI-capable share's coefficient is $-0.24$ ($p=0.09$), similar to employment weighting. The 2019 and 2025 estimates use non-overlapping survey panels, so no covariance term is required for this calculation. These exercises, with the leave-one-family-out bounds reported in the note to Table~\ref{tab:weights}, establish influence robustness for the between-family association; they do not overcome the null within-family fixed-effects estimate.

\subsection{The three-window decomposition and sample definition}

\begin{table}[H]
\centering
\caption{Window decomposition: characteristic coefficients by window}
\label{tab:dissoc}
\small
\begin{tabular}{lccc}
\toprule
 & 2019--2022 & 2022--2024 & 2024--2025 \\
\midrule
\multicolumn{4}{l}{\emph{A. Univariate (one characteristic at a time)}}\\
AI-capable share $a_j$ & $-0.09$       & $-0.14^{**}$  & $-0.06$ \\
                       & $(0.08)$      & $(0.06)$      & $(0.04)$ \\
Execution share $e_j$  & $-0.24^{**}$  & $-0.27^{***}$ & $-0.12^{**}$ \\
                       & $(0.11)$      & $(0.09)$      & $(0.05)$ \\
\addlinespace
\multicolumn{4}{l}{\emph{B. Joint (both entered together)}}\\
AI-capable share $a_j$ & $-0.16^{+}$   & $-0.22^{***}$ & $-0.09^{**}$ \\
                       & $(0.08)$      & $(0.06)$      & $(0.04)$ \\
Execution share $e_j$  & $-0.32^{***}$ & $-0.37^{***}$ & $-0.17^{***}$ \\
                       & $(0.11)$      & $(0.10)$      & $(0.06)$ \\
\bottomrule
\end{tabular}
\par\vspace{4pt}
\footnotesize\raggedright Note: Each cell is an OLS coefficient for the windowed log employment change on a fixed sample of $336$ cognitive occupations; robust standard errors are in parentheses. Panel A enters the characteristics separately and Panel B jointly. Because the windows partition 2019--2025 on the same sample, the Panel-B coefficients sum approximately to the full-period slopes in Table~\ref{tab:reg}, column 2.
\par\smallskip The table is descriptive. Formal tests using annualized slopes do not reject equality across windows for execution share or the 2026-vintage capability score. The vintage-valid Eloundou results appear in the text and Table~\ref{tab:pretrend}. $^{+}\,p<0.10$, $^{**}\,p<0.05$, $^{***}\,p<0.01$.
\end{table}

Table~\ref{tab:dissoc} reports the window decomposition on the fixed sample. Formal annualized slope tests do not reject a break for either the execution share or the AI-capable share. The execution slopes are $-0.08$, $-0.13$, and $-0.12$ per year; the diffusion-versus-pre comparison has $p=0.35$, and equality of all three has $p=0.61$. The AI-capable share's slopes are $-0.03$, $-0.07$, and $-0.06$, with pairwise and joint $p$-values between $0.35$ and $0.64$. The Eloundou exposure slopes on this fixed sample are $-0.052$ (s.e.\ $0.018$) per year in 2019--2022, $-0.071$ ($0.023$) in 2022--2024, and $-0.123$ ($0.037$) in 2024--2025. They are significantly negative in every window, equality across the three is marginal ($p=0.079$), and the 2024--2025 slope is steeper than 2019--2022 by $0.070$ ($p=0.041$). This corroborates, on a different window partition, the steepening that Table~\ref{tab:pretrend} detects on the estimator-consistent panel.

Defining cognitive occupations as the contiguous SOC range 11 through 43 adds $117$ healthcare-support, protective-service, food-preparation, cleaning, personal-care, and sales occupations. Their mean 2019--2022 employment change is $-5.6$ percent, compared with $-0.5$ percent in the ten white-collar majors. The broader sample then produces execution coefficients of $-0.30$, $-0.06$, and $-0.09$ and capability coefficients of $+0.07$, $-0.17$, and $-0.05$ across the three windows. This ``double dissociation'' is driven by adding pandemic-sensitive in-person service occupations. It is a fact about a different population, not a timing result within white-collar work.

\subsection{Within-occupation adoption estimates}
\label{app:adoption}

Table~\ref{tab:adoption} reports the estimates behind Section~\ref{sec:adoption}. Each coefficient comes from a separate regression of an occupation-industry cell's annualized log employment change on the interactions in equation~\eqref{eq:adoption}. The row names the fixed effects.

\begin{table}[!htbp]
\centering
\caption{The adoption gradient: execution share times predetermined industry AI adoption}
\label{tab:adoption}
\footnotesize
\setlength{\tabcolsep}{3pt}
\begin{tabular}{@{}lccccc@{}}
\toprule
 & 2016--19 & 2019--22 & 2022--25 & 2019--25 & 2022--25 minus \\
 & (pre-period) & (pandemic) & & (six years) & 2016--19 \\
\midrule
\multicolumn{6}{@{}l}{\emph{Panel A. National execution slope per SD, cross-industry (scale)}} \\
Execution share & ${\iaNatBPre}$ & ${\iaNatBPan}$ & ${\iaNatBPost}$ & ${\iaNatBSix}$ &  \\
 & {\scriptsize (${\iaNatSPre}$)} & {\scriptsize (${\iaNatSPan}$)} & {\scriptsize (${\iaNatSPost}$)} & {\scriptsize (${\iaNatSSix}$)} &  \\
\addlinespace
\multicolumn{6}{@{}l}{\emph{Panel B. Execution $\times$ sector AI use; occupation and sector fixed effects}} \\
Share of firms using AI, ABS & ${\iaFirmBPre}$ & ${\iaFirmBPan}$ & ${\iaFirmBPost}$ & ${\iaFirmBSix}$ & ${\iaFirmBDiff}$ \\
 & {\scriptsize (${\iaFirmSPre}$) [${\iaFirmPPre}$]} & {\scriptsize (${\iaFirmSPan}$) [${\iaFirmPPan}$]} & {\scriptsize (${\iaFirmSPost}$) [${\iaFirmPPost}$]} & {\scriptsize (${\iaFirmSSix}$) [${\iaFirmPSix}$]} & {\scriptsize (${\iaFirmSDiff}$) [${\iaFirmPDiff}$]} \\
Employment share in AI-using firms, ABS & ${\iaEmpBPre}$ & ${\iaEmpBPan}$ & ${\iaEmpBPost}$ & ${\iaEmpBSix}$ & ${\iaEmpBDiff}$ \\
 & {\scriptsize (${\iaEmpSPre}$) [${\iaEmpPPre}$]} & {\scriptsize (${\iaEmpSPan}$) [${\iaEmpPPan}$]} & {\scriptsize (${\iaEmpSPost}$) [${\iaEmpPPost}$]} & {\scriptsize (${\iaEmpSSix}$) [${\iaEmpPSix}$]} & {\scriptsize (${\iaEmpSDiff}$) [${\iaEmpPDiff}$]} \\
Share of firms using AI, BTOS (post) & ${\iaBtosBPre}$ & ${\iaBtosBPan}$ & ${\iaBtosBPost}$ & ${\iaBtosBSix}$ &  \\
 & {\scriptsize (${\iaBtosSPre}$) [${\iaBtosPPre}$]} & {\scriptsize (${\iaBtosSPan}$) [${\iaBtosPPan}$]} & {\scriptsize (${\iaBtosSPost}$) [${\iaBtosPPost}$]} & {\scriptsize (${\iaBtosSSix}$) [${\iaBtosPSix}$]} &  \\
\addlinespace
\multicolumn{6}{@{}l}{\emph{Panel C. Horse race of ABS share-of-firms measures, one regression per window}} \\
Execution $\times$ AI & ${\iaHrAIBPre}$ & ${\iaHrAIBPan}$ & ${\iaHrAIBPost}$ & ${\iaHrAIBSix}$ &  \\
 & {\scriptsize (${\iaHrAISPre}$) [${\iaHrAIPPre}$]} & {\scriptsize (${\iaHrAISPan}$) [${\iaHrAIPPan}$]} & {\scriptsize (${\iaHrAISPost}$) [${\iaHrAIPPost}$]} & {\scriptsize (${\iaHrAISSix}$) [${\iaHrAIPSix}$]} &  \\
Execution $\times$ cloud computing & ${\iaHrCloudBPre}$ & ${\iaHrCloudBPan}$ & ${\iaHrCloudBPost}$ & ${\iaHrCloudBSix}$ &  \\
 & {\scriptsize (${\iaHrCloudSPre}$) [${\iaHrCloudPPre}$]} & {\scriptsize (${\iaHrCloudSPan}$) [${\iaHrCloudPPan}$]} & {\scriptsize (${\iaHrCloudSPost}$) [${\iaHrCloudPPost}$]} & {\scriptsize (${\iaHrCloudSSix}$) [${\iaHrCloudPSix}$]} &  \\
Execution $\times$ specialized software & ${\iaHrSoftBPre}$ & ${\iaHrSoftBPan}$ & ${\iaHrSoftBPost}$ & ${\iaHrSoftBSix}$ &  \\
 & {\scriptsize (${\iaHrSoftSPre}$) [${\iaHrSoftPPre}$]} & {\scriptsize (${\iaHrSoftSPan}$) [${\iaHrSoftPPan}$]} & {\scriptsize (${\iaHrSoftSPost}$) [${\iaHrSoftPPost}$]} & {\scriptsize (${\iaHrSoftSSix}$) [${\iaHrSoftPSix}$]} &  \\
\addlinespace
\multicolumn{6}{@{}l}{\emph{Panel D. Execution $\times$ AI-specialist staffing, May 2016, percentile rank; $\iaFineUnits$ industries}} \\
Occupation and industry FE & ${\iaFineOccBPre}$ & ${\iaFineOccBPan}$ & ${\iaFineOccBPost}$ & ${\iaFineOccBSix}$ & ${\iaFineOccBDiff}$ \\
 & {\scriptsize (${\iaFineOccSPre}$) [${\iaFineOccPPre}$]} & {\scriptsize (${\iaFineOccSPan}$) [${\iaFineOccPPan}$]} & {\scriptsize (${\iaFineOccSPost}$) [${\iaFineOccPPost}$]} & {\scriptsize (${\iaFineOccSSix}$) [${\iaFineOccPSix}$]} & {\scriptsize (${\iaFineOccSDiff}$) [${\iaFineOccPDiff}$]} \\
Occupation-by-sector and industry FE & ${\iaFineSecBPre}$ & ${\iaFineSecBPan}$ & ${\iaFineSecBPost}$ & ${\iaFineSecBSix}$ & ${\iaFineSecBDiff}$ \\
 & {\scriptsize (${\iaFineSecSPre}$) [${\iaFineSecPPre}$]} & {\scriptsize (${\iaFineSecSPan}$) [${\iaFineSecPPan}$]} & {\scriptsize (${\iaFineSecSPost}$) [${\iaFineSecPPost}$]} & {\scriptsize (${\iaFineSecSSix}$) [${\iaFineSecPSix}$]} & {\scriptsize (${\iaFineSecSDiff}$) [${\iaFineSecPDiff}$]} \\
\multicolumn{6}{@{}l}{Occupation and industry FE, with computer staffing:} \\
\quad AI staffing &${\iaFineCompAIBPre}$ & ${\iaFineCompAIBPan}$ & ${\iaFineCompAIBPost}$ & ${\iaFineCompAIBSix}$ & ${\iaFineCompAIBDiff}$ \\
 & {\scriptsize (${\iaFineCompAISPre}$) [${\iaFineCompAIPPre}$]} & {\scriptsize (${\iaFineCompAISPan}$) [${\iaFineCompAIPPan}$]} & {\scriptsize (${\iaFineCompAISPost}$) [${\iaFineCompAIPPost}$]} & {\scriptsize (${\iaFineCompAISSix}$) [${\iaFineCompAIPSix}$]} & {\scriptsize (${\iaFineCompAISDiff}$) [${\iaFineCompAIPDiff}$]} \\
\quad computer staffing & ${\iaFineCompBPre}$ & ${\iaFineCompBPan}$ & ${\iaFineCompBPost}$ & ${\iaFineCompBSix}$ & ${\iaFineCompBDiff}$ \\
 & {\scriptsize (${\iaFineCompSPre}$) [${\iaFineCompPPre}$]} & {\scriptsize (${\iaFineCompSPan}$) [${\iaFineCompPPan}$]} & {\scriptsize (${\iaFineCompSPost}$) [${\iaFineCompPPost}$]} & {\scriptsize (${\iaFineCompSSix}$) [${\iaFineCompPSix}$]} & {\scriptsize (${\iaFineCompSDiff}$) [${\iaFineCompPDiff}$]} \\
\multicolumn{6}{@{}l}{BTOS share of firms using AI by three-digit industry (post-outcome):} \\
\quad occupation and industry FE & ${\iaFineBtosOccBPre}$ & ${\iaFineBtosOccBPan}$ & ${\iaFineBtosOccBPost}$ & ${\iaFineBtosOccBSix}$ &  \\
 & {\scriptsize (${\iaFineBtosOccSPre}$) [${\iaFineBtosOccPPre}$]} & {\scriptsize (${\iaFineBtosOccSPan}$) [${\iaFineBtosOccPPan}$]} & {\scriptsize (${\iaFineBtosOccSPost}$) [${\iaFineBtosOccPPost}$]} & {\scriptsize (${\iaFineBtosOccSSix}$) [${\iaFineBtosOccPSix}$]} &  \\
\quad occupation-by-sector and industry FE & ${\iaFineBtosSecBPre}$ & ${\iaFineBtosSecBPan}$ & ${\iaFineBtosSecBPost}$ & ${\iaFineBtosSecBSix}$ &  \\
 & {\scriptsize (${\iaFineBtosSecSPre}$) [${\iaFineBtosSecPPre}$]} & {\scriptsize (${\iaFineBtosSecSPan}$) [${\iaFineBtosSecPPan}$]} & {\scriptsize (${\iaFineBtosSecSPost}$) [${\iaFineBtosSecPPost}$]} & {\scriptsize (${\iaFineBtosSecSSix}$) [${\iaFineBtosSecPSix}$]} &  \\
\midrule
Cells, sector design (Panels B, C) & $\iaCellsPre$ & $\iaCellsPan$ & $\iaCellsPost$ & $\iaCellsSix$ & \\
Cells, industry design (Panel D) & $\iaFineCellsPre$ & $\iaFineCellsPan$ & $\iaFineCellsPost$ & $\iaFineCellsSix$ & \\
Occupations, sector / industry design & $\iaOccPre$ / $\iaFineNoccPre$ & $\iaOccPan$ / $\iaFineNoccPan$ & $\iaOccPost$ / $\iaFineNoccPost$ & $\iaOccSix$ / $\iaFineNoccSix$ & \\
\bottomrule
\end{tabular}
\par\vspace{4pt}
\footnotesize\raggedright Note: OLS on occupation-industry cells from national OEWS industry-specific estimates; the outcome is the annualized log employment change. Coefficients and standard errors are multiplied by 100 and read as percent per year for a one-standard-deviation difference in execution share times a one-standard-deviation difference in adoption. Cells are cognitive occupations with an index score, positive employment at both ends, and at least two industries; Panel D drops the AI-specialist occupations. The 2016--19 window uses MB3 research estimates and one-to-one SOC-2010 codes; 2019--22 runs from the 2019 MB3 estimate to the published 2022 estimate. ABS: 2018 Annual Business Survey, reference year 2017, $\iaNSectors$ private sectors, low, moderate or high AI use in producing goods or services. BTOS rows postdate the outcomes and are a ceiling check; Panel D's cover $\iaFineBtosUnits$ industries. Panel D industries are four-digit NAICS codes published identically in all seven vintages, three-digit codes otherwise, and whole sectors for retail and information.
\par\smallskip Standard errors in parentheses are two-way clustered by occupation and by sector (Panels B, C) or industry (Panel D). Bracketed $p$-values are from a wild cluster restricted bootstrap by sector with Rademacher weights: all $2^{19}$ sign vectors for single-window sector regressions, 9{,}999 draws otherwise. The last column is the post-minus-pre coefficient from one regression stacking the two windows with window-specific fixed effects.
\end{table}

\section{Scoring Rubrics and Replication Protocol}
\label{app:rubric}

\subsection*{The execution-evaluation rubric}

Each O*NET task statement receives an execution score from $0$ to $100$ under the rubric summarized here; the full system prompt is included in the replication materials. Execution means producing or performing: the worker generates an output or carries out an action, such as making, building, operating, drafting, calculating, assembling, installing, entering, or preparing. Evaluation means judging or deciding: the worker assesses the correctness, quality, or fit of an output or situation and renders a verdict. Examples are checking, reviewing, inspecting, diagnosing, approving, prioritizing, choosing, auditing, and supervising.

A score of $100$ denotes pure execution, $0$ pure evaluation, and $50$ an even mix. The model scores the dominant cognitive act in the task as written and sees only the statement, not the occupation title.

The rubric treats routineness as a separate dimension. It explicitly warns against equating execution with routine work or evaluation with non-routine work. All four combinations are possible. Performing a complex one-off surgery or composing an original score is non-routine execution and so scores high. Inspecting each unit against a specification or screening applications against fixed criteria is routine evaluation and scores low. Data entry is routine execution, whereas diagnosing an unusual fault is non-routine evaluation. Table~\ref{tab:rubric} reports representative examples. This crossed design separates the measure from a leading-verb heuristic and explains why it can correlate with routine intensity without collapsing into it.

\subsection*{The AI-capability rubric}

The same task statements are scored separately on a $0$ to $100$ AI-capability scale. A score of $100$ means that a current frontier language model with typical tool access can perform the task as stated, end to end, at a professional standard. A score of $0$ means that it cannot. Intermediate scores indicate partial completion or substantial required human involvement.

The rubric asks about technical feasibility, not legality, licensing, accountability, or cost. Capability is evaluated relative to the model frontier at the date of the scoring run, so the measure has an explicit vintage. As in the execution run, occupation titles are hidden. The worked examples deliberately cross the two axes: AI-capable execution, AI-capable evaluation, AI-incapable execution, and AI-incapable evaluation. Table~\ref{tab:rubric_ai} gives actual O*NET statements from each quadrant.

\subsection*{Protocol, common to both runs}

The baseline measures use O*NET version 26.1 (November 2021); robustness measures use version 30.3 (May 2026). The production model is Claude Sonnet 4.6 (API identifier \texttt{claude-sonnet-4-6}). Execution scoring was completed on June 5, 2026; capability scoring was completed on June 11, 2026. The baseline run scores $19{,}265$ statements, and the version-30.3 rebuild scores $18{,}796$, at temperature zero in batches of $25$. Occupation titles are hidden, statements appear in fixed O*NET order, and malformed responses trigger up to five retries with exponential backoff. Execution and capability are scored in separate runs under separate system prompts, so neither score is an input into the other. The prompts, batching and parsing code, capability-scoring log, and complete score files are included in the replication archive.

The model-coder audit described in Section~\ref{sec:data} uses a stratified, occupation-blind sample of $500$ statements, drawn with a fixed seed and oversampling supervisory, diagnostic, reviewing, and mixed tasks. Two model coders from a different generation independently score the tasks under different framings while blind to occupations, production scores, and each other. Their scores correlate $0.991$, and their consensus correlates $0.924$ with the production scores without level bias. Two protocol limits should be recorded. Agreement is computed unweighted over a draw that deliberately oversamples hard cases, so it is not a population-weighted validation statistic. And the exact model identifiers of the two audit coders were not logged at run time. The archived scores and the fixed seed reproduce the comparison, but the specific coder versions cannot be recovered. The $30$ statements on which the audit coders agree with one another but differ materially from production remain a diagnostic disagreement set rather than the basis of the human benchmark.

The human audit is a separate fixed-seed simple random sample without replacement of $200$ statements from the full O*NET 26.1 task universe. It covers $180$ occupations; no occupation contributes more than four statements. Four independent coders score the statements while blind to occupation titles, model scores, and one another's responses. The first two coders use English forms. The later two use bilingual Chinese--English forms, with the English statement controlling if the two versions appear to differ. The benchmark is the unadjudicated mean of all four coders. The production score correlates $0.806$ with that mean and differs from it by $12.3$ points on average; the corresponding three-way agreement is $75.5$ percent. Single-coder absolute agreement is moderate, ICC$(A,1)=0.482$, while ICC$(A,4)=0.788$ for the four-coder mean. These results support task-level construct validity while also documenting substantial scale-use differences across people.

The archive contains the blind and keyed model-audit files, all four original human-coder workbooks, the merged task-level scores, validation and bootstrap summaries, and the analysis script that reproduces Table~\ref{tab:humanaudit}. The occupation-level comparisons with O*NET evaluation activities and the Anthropic Economic Index are reported in Section~\ref{sec:data}. So I describe the execution score as human-benchmarked at the task level and externally anchored at the occupation level. Neither benchmark makes the human mean ground truth or validates the occupation-level employment design.

\begin{table}[ht]
\centering
\caption{Worked examples from the scoring rubric}
\label{tab:rubric}
\small
\begin{tabular}{lcc}
\toprule
Task statement & Routine? & Score \\
\midrule
Enter patient demographic data into the record system & routine & 95 \\
Perform a complex, non-standard cardiac surgery & non-routine & 90 \\
Compose original background music for an advertisement & non-routine & 88 \\
Inspect each weld against the engineering spec and flag defects & routine & 12 \\
Screen loan applications against criteria and approve or deny & routine & 10 \\
Decide which grant proposals to fund on scientific merit & non-routine & 5 \\
\bottomrule
\end{tabular}
\par\vspace{4pt}
\footnotesize\raggedright Note: the execution score runs from $0$ (pure evaluation) to $100$ (pure execution). Non-routine execution scores high and routine evaluation scores low, so the score does not track routineness.
\end{table}

\begin{table}[ht]
\centering
\caption{AI-capability worked examples: the execution and AI-capability axes cross}
\label{tab:rubric_ai}
\small
\begin{tabular}{lcc}
\toprule
O*NET task statement & Execution $e_j$ & AI-capable $a_j$ \\
\midrule
\multicolumn{3}{l}{\emph{AI-capable, execution-weighted}}\\
Enter tax return information into computers for processing & 100 & 88 \\
Compute payment schedules & 95 & 95 \\
\midrule
\multicolumn{3}{l}{\emph{AI-capable, evaluation-weighted}}\\
Check completed work for spelling, grammar, punctuation, and format & 10 & 90 \\
Review manuscripts for publication in books and professional journals & 5 & 82 \\
\midrule
\multicolumn{3}{l}{\emph{AI-incapable, execution-weighted}}\\
Climb ladders to reach aircraft surfaces to be cleaned & 100 & 1 \\
Clean shelves, counters, and tables & 100 & 2 \\
\midrule
\multicolumn{3}{l}{\emph{AI-incapable, evaluation-weighted}}\\
Inspect finished dies for smoothness, contour conformity, and defects & 5 & 6 \\
Observe nurses and visit patients to ensure proper nursing care & 10 & 5 \\
\bottomrule
\end{tabular}
\par\vspace{4pt}
\footnotesize\raggedright Note: actual O*NET version 26.1 task statements with their language-model scores. The execution score $e_j$ runs from $0$ (pure evaluation) to $100$ (pure execution); the AI-capability score $a_j$ from $0$ (the model cannot perform the task as stated) to $100$ (it can, end to end). All four execution$\times$AI-capability combinations are populated, so the AI-capable share is not a relabeling of the execution share.
\end{table}

\section{CPS Seniority Specifications}
\label{app:cps}

Tables~\ref{tab:cps} and \ref{tab:cpsfe} report the person-level specifications summarized in Section~\ref{sec:results}. The outcome is employment among labor-force participants in cognitive occupations, and the coefficient of interest is the entry-age-by-post-2022 triple interaction. The sample contains $185{,}061$ person-years, weighted by the ASEC supplement weight, with standard errors clustered by occupation. Adding occupation-by-entry or entry-by-year fixed effects leaves the triple interaction nearly unchanged, so the null is not driven by broad occupation composition or common year shocks.

The outcome omits several margins on which the mechanism may first appear. Conditioning on labor-force participation excludes exit from the labor force and misses workers who never enter an occupation. The age comparison, 22--30 versus 35--64, is broader than job seniority or tenure. Occupation for unemployed workers refers to the last job, and annual ASEC observations blur timing.

The null is nevertheless informative about effect size. The standard error of the execution triple is $0.0030$, implying a minimum detectable effect of roughly $0.008$, or $0.8$ percentage points, under conventional power calculations. On an outcome with mean $0.976$, the design rules out effects of that size or larger on this margin. It remains uninformative about smaller effects and about entry hiring. A stronger public-data design would use matched Basic Monthly CPS rotations to study employment entry, unemployment, occupation switching, and labor-force exit. It could also construct occupation-by-age employment-to-population ratios from the CPS or ACS. Payroll, vacancy, resume, or job-posting data with seniority labels would provide the more direct test described in Section~\ref{sec:conclusion}.

\begin{table}[H]
\centering
\caption{CPS seniority test: employment triple interactions, cognitive occupations}
\label{tab:cps}
\small
\begin{tabular}{lccc}
\toprule
 & (1) & (2) & (3) \\
 & execution & AI-capable & both \\
\midrule
Execution $\times$ entry $\times$ post  & $0.0023$   &            & $0.0025$ \\
                                        & $(0.0030)$ &            & $(0.0029)$ \\
AI-capable $\times$ entry $\times$ post &            & $0.0012$   & $0.0019$ \\
                                        &            & $(0.0048)$ & $(0.0048)$ \\
Execution $\times$ post                 & $-0.0007$  &            & $-0.0008$ \\
                                        & $(0.0019)$ &            & $(0.0019)$ \\
AI-capable $\times$ post                &            & $-0.0012$  & $-0.0014$ \\
                                        &            & $(0.0023)$ & $(0.0023)$ \\
Execution $\times$ entry                & $0.0009$   &            & $0.0005$ \\
                                        & $(0.0021)$ &            & $(0.0020)$ \\
AI-capable $\times$ entry               &            & $-0.0026$  & $-0.0026$ \\
                                        &            & $(0.0032)$ & $(0.0029)$ \\
\midrule
Observations (person-years)             & 185{,}061  & 185{,}061  & 185{,}061 \\
Mean of employed                        & 0.976      & 0.976      & 0.976 \\
Occupation clusters                     & 191        & 191        & 191 \\
\bottomrule
\end{tabular}
\par\vspace{4pt}
\footnotesize\raggedright Note: Linear probability models for employment among CPS ASEC labor-force participants in cognitive occupations, 2019--2025. Entry ages are 22--30; the comparison group is 35--64. Estimates use ASEC weights, standardized task measures, and occupation-clustered standard errors.
\par\smallskip Lower-order entry and post terms are included but not shown. This table has no occupation or year fixed effects; Table~\ref{tab:cpsfe} adds them. The entry-by-post triple is the P2 test. It is small, insignificant, and slightly positive for both characteristics and in the joint specification. $^{+}\,p<0.10$, $^{**}\,p<0.05$, $^{***}\,p<0.01$.
\end{table}

\begin{table}[H]
\centering
\caption{CPS seniority triple under fixed effects, cognitive occupations}
\label{tab:cpsfe}
\small
\begin{tabular}{lcccc}
\toprule
 & (1) & (2) & (3) & (4) \\
 & no FE & occ, year & occ$\times$entry, year & occ, year, entry$\times$year \\
\midrule
Execution $\times$ entry $\times$ post  & $0.0023$   & $0.0030$   & $0.0030$   & $0.0029$ \\
                                        & $(0.0030)$ & $(0.0030)$ & $(0.0029)$ & $(0.0030)$ \\
AI-capable $\times$ entry $\times$ post & $0.0012$   & $0.0010$   & $0.0017$   &          \\
                                        & $(0.0048)$ & $(0.0047)$ & $(0.0046)$ &          \\
\midrule
Occupation FE                & no & yes & --  & yes \\
Year FE                      & no & yes & yes & yes \\
Occupation $\times$ entry FE & no & no  & yes & no  \\
Entry $\times$ year FE       & no & no  & no  & yes \\
Observations                 & 185{,}061 & 185{,}061 & 185{,}060 & 185{,}061 \\
\bottomrule
\end{tabular}
\par\vspace{4pt}
\footnotesize\raggedright Note: Each cell is the entry-by-post-by-characteristic coefficient from a separate CPS ASEC linear probability model. The sample contains cognitive labor-force participants aged 22--30 and 35--64 in 2019--2025. Estimates use ASEC weights, standardized task measures, lower-order interactions, and occupation-clustered standard errors ($191$ clusters).
\par\smallskip Column 1 repeats Table~\ref{tab:cps}; column 3 absorbs the occupation main effect within occupation-by-entry fixed effects. The AI-capable triple is reported for the fixed-effect designs in which it was estimated. The triple is small, positive, and insignificant throughout. $^{+}\,p<0.10$, $^{**}\,p<0.05$, $^{***}\,p<0.01$.
\end{table}

\section{A Quantitative Scenario Exercise}
\label{sec:estimation}

This appendix reports a scenario, not an estimate. I do not attempt indirect inference because the required moments are absent. The cognitive sample does not contain a post-2022 execution trend that is steeper than the pre-2022 trend, and the technology elasticities would have to be calibrated rather than estimated. So an auxiliary-model exercise would add little beyond internal consistency. The useful output of the model is its timing logic.

The scenario applies the pipeline in Appendix~\ref{app:formal} with cohort lags $L\in\{6,10,15\}$ years, a positive post-2022 practice-deficit path in the range allowed by the reduced-form evidence, and standard depreciation. Conditional on that deficit, three timing implications are robust. First, the senior evaluation stock is unchanged until the first AI-era cohort reaches seniority. Divergence begins near $2022+L$: 2028 at the shortest lag and the early 2030s at the central lag. Second, the full 2022--2025 practice deficit reaches the senior stock near $2025+L$, from 2031 to 2040 across the three lags. Third, any downstream deterioration in evaluation quality follows with an additional lag.

These dates are cohort arithmetic and are relatively insensitive to the shock magnitude or depreciation. The size of the decline is not. Figure~\ref{fig:predictions} plots the central-lag path, and no magnitude in either figure should be interpreted as an estimate.

\begin{figure}[H]
\centering
\includegraphics[width=0.85\linewidth]{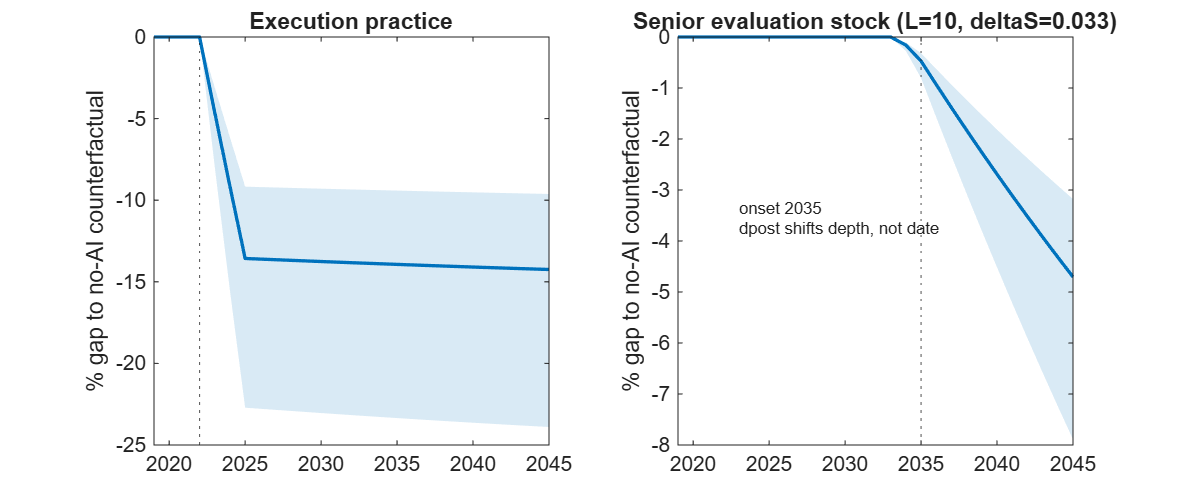}
\caption{Illustrative scenario for execution practice and the senior evaluation stock at $L=10$. Divergence begins near $2022+L$, and the full 2022--2025 practice deficit reaches the senior stock near $2025+L$. Timing follows from cohort arithmetic; magnitudes are not estimates.}
\label{fig:predictions}
\end{figure}

Figure~\ref{fig:lagsweep} holds the same parameterization fixed and varies the cohort lag over $L\in\{6,10,15\}$. The lag shifts the onset of the senior decline but not the common long-run asymptote in this illustration.

\begin{figure}[H]
\centering
\includegraphics[width=0.62\textwidth]{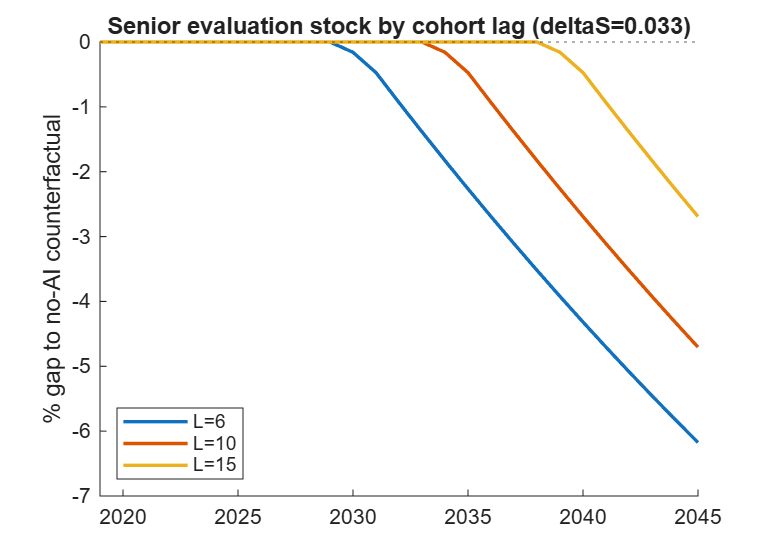}
\caption{Illustrative senior-stock paths for $L\in\{6,10,15\}$ with $\delta_S=0.033$. The full 2022--2025 practice deficit becomes visible near 2031, 2035, and 2040. The lag shifts the onset; the magnitude is calibrated rather than estimated.}
\label{fig:lagsweep}
\end{figure}

The empirical signature is a sequence rather than a single coefficient. Human execution practice falls first; the senior evaluation stock remains stable until the affected cohorts mature; and evaluation quality deteriorates later. The concluding section identifies the data each stage would require.

\end{document}